\documentclass[times,sort&compress,3p]{elsarticle}
\journal{Journal of Multivariate Analysis}
\usepackage[labelfont=bf]{caption}

\usepackage{amsmath,amsfonts,amssymb,amsthm,booktabs,color,epsfig,graphicx,url}

\usepackage{bm,bbm,enumitem,multirow}
\usepackage{cancel,comment}
\usepackage{algorithm,algpseudocode}
\usepackage[usenames,dvipsnames]{xcolor}

\theoremstyle{plain}%
\newtheorem{thm}{Theorem}
\newtheorem*{othersthm}{Theorem}
\newtheorem{prop}{Proposition}
\newtheorem{lem}{Lemma}
\newtheorem{cor}{Corollary}

\theoremstyle{definition}
\newtheorem{defn}{Definition}

\newcommand{\SFS}[1]{{#1}}

\newcommand{\Real}[1]{\mathbb{R}^{#1}}

\newcommand{\Variance}[1]{\mathbb{V}(#1)}
\newcommand{\diagonalmatrix}[1]{\text{diag}(#1)}

\def\squarematrixsymbol{{\scalebox{0.5}{$\,\square$}}}
\newcommand{\yv}{{y}}

\newcommand{\dimy}{m}
\newcommand{\nfactrue}{r}
\newcommand{\nfac}{k}
\newcommand{\Sigmaeps}{\Sigma_0}
\newcommand{\Sigmak}[1]{\Sigma_{#1}}

\newcommand{\Sy}{S_y}
\newcommand{\Msp}{M}

\newcommand{\sparsespace}{\mathcal{B}}%
\newcommand{\sparsespacerem}{{\sparsespace^\text{rem}}} %
\newcommand{\binarymatrix}{{\delta}}
\newcommand{\facload}[1]{\beta_{#1}}
\newcommand{\facloadmatrix}{{\beta}}
\newcommand{\facloadtrue}{{\Lambda}}
\newcommand{\squarefacloadmatrix}{\facloadmatrix^\squarematrixsymbol}
\newcommand{\squarebinarymatrix}{\binarymatrix^\squarematrixsymbol}

\newcommand{\binarymatrixrem}{{\binarymatrix}^\text{rem}}
\newcommand{\facloadmatrixrem}{{\facloadmatrix}^\text{rem}}

\newcommand{\vergraph}{B}
\newcommand{\vercolvertices}{V_\text{col}}
\newcommand{\verrowvertices}{V_\text{row}}
\newcommand{\veredges}{E_B}
\newcommand{\rowdeletion}[2]{\text{RD}({#1},{#2})}
\newcommand{\countingrule}[2]{\text{CR}({#1},{#2})}
\newcommand{\RDr}{\rowdeletion{r}{1}}
\newcommand{\RDrs}{\rowdeletion{r}{s}}
\newcommand{\RDro}{\rowdeletion{r}{0}}
\newcommand{\CRr}{\countingrule{r}{1}}
\newcommand{\CRrs}{\countingrule{r}{s}}
\newcommand{\CRro}{\countingrule{r}{0}}

\newcommand{\STARTNEW}{\textcolor{red}{\textbf{Start NEW}}}
\newcommand{\ENDNEW}{\textcolor{red}{\textbf{End NEW}}}

\newcommand{\mycomment}[1]{}

\usepackage{hyperref}

\usepackage{tikz}
\usetikzlibrary{positioning,shapes.geometric,shapes.misc,chains,fit,shapes,calc,decorations.pathmorphing,matrix,decorations.pathreplacing}

\def\stepsize{3mm}
\colorlet{covercolor}{green!80!black}

\tikzset{nonzero/.style={draw, circle, color=white, fill=blue!90!black, inner sep=0.5mm}}
\tikzset{leading/.style={draw, isosceles triangle, isosceles triangle apex angle=60, color=white, fill=red!90!black, inner sep=0.5mm, rotate=-30}}
\tikzset{axis label/.style={}}
\tikzset{matrix box/.style={ultra thin, color=gray}}
\tikzset{venn diagram/.style={draw, rounded rectangle, inner sep=0pt}}
\tikzset{edge/.style={}}

\begin{document}

\begin{frontmatter}

  \title{Cover It Up! Bipartite Graphs Uncover Identifiability in Sparse Factor Analysis}

  \author[1]{Darjus Hosszejni}
  \author[2]{Sylvia Fr\"uhwirth-Schnatter}

  \address[1]{Department of Finance, Accounting, and Statistics, WU Vienna University of Economics and Business, Austria}
  \address[2]{Department of Finance, Accounting, and Statistics, WU Vienna University of Economics and Business, Austria}

  \cortext[mycorrespondingauthor]{Corresponding author. Email address: \url{darjus.hosszejni@wu.ac.at}}

\begin{abstract}
    \SFS{Factor models are an indispensable tool in dimension reduction in multivariate statistical analysis. Despite their popularity,} little attention has been given to formally address identifiability of these models beyond standard rotation-based identification. To fill this gap, \SFS{the present paper focuses
   on uniquely identifying the variance decomposition in the factor representation without imposing any constraints or structure on the loading matrix. We rely on a counting rule for the zero-nonzero pattern of the loading matrix and prove sufficiency of this condition for achieving  variance identification. The proof is based on connecting factor analysis with some classical elements from graph and network theory.}
    Furthermore, we provide a computationally efficient tool for verifying the counting rule.
    Our methodology is illustrated for simulated as well as real data in the context of post-processing posterior draws in sparse  Bayesian factor analysis.
  \end{abstract}


\end{frontmatter}

\section{Introduction\label{sec:intro}}

A popular technique of dimension reduction in multivariate analysis is principal component analysis (PCA) which relies on the single value decomposition (SVD) of the sample covariance matrix
$\Sy$ of realizations  of an $\dimy$-dimensional random variable $Y$, see e.g. \cite{and:int}. More specifically,
only the eigenvectors
$U_1$ corresponding to the $\nfactrue$ largest eigenvalues $D_1$ are kept, while the variation explained by the eigenvectors
$U_2$ corresponding to the remaining eigenvalues $D_2$ is ignored, i.e.
\begin{equation} \label{SVD}
    \Sy= U_1 D_1 U_1^ \top + U_2 D_2 U_2^ \top
    \approx  \facloadtrue \facloadtrue  ^\top
\end{equation}
where $\facloadtrue= U_1 D_1^{1/2}$ is a $\dimy \times \nfactrue$ matrix, typically with  $\nfactrue \ll \dimy$. In its original form, PCA is purely a data reduction technique without much insight into the data generating process.

A statistical modelling framework derived from PCA is  probabilistic PCA  (\cite{tip-bis:pro}) which adds random noise to account for the unexplained variance due to dropping the term $U_2 D_2 U_2^ \top$ in \eqref{SVD}. Assuming w.o.l.g. that $Y$ is centered, %
$Y \sim N(0,\Omega)$ is assumed to arise from a zero-mean Gaussian distribution with covariance matrix $\Omega= \facloadtrue \facloadtrue  ^\top + \sigma^2 I_{\dimy}$. In this model, the covariance between the components $Y_i$ and $Y_\ell$  of $Y$ is
explained by the inner product of row  $i$ and $\ell$ of $\facloadtrue$, while a single parameter,  $\sigma^2$, is present to control the fraction of unexplained variance, $\sigma^2/\Omega_{ii}$, for all components of $Y$.  Considering for illustration a random variable $Y$ which is not only centered, but also standardized (i.e. $\Omega_{ii}=1$), it becomes apparent that probabilistic PCA relies on the rather strict assumption that the fraction of unexplained variance is the same for all components of $Y$.

More flexibility in this regard is obtained by the multi-factor model introduced by \cite{thu:vec} which found numerous applications in applied multivariate analysis and will be the focus of the present paper. The model introduces an idiosyncratic variance $\sigma^2 _i$ to account for the unexplained variance of each components $Y_i$ of $Y$ and decomposes the covariance matrix $\Omega$ as
\begin{equation} \label{Omegatrue}
    \Omega=\facloadtrue \facloadtrue ^\top + \Sigmaeps,
\end{equation}
where $\facloadtrue$ is the $\dimy \times \nfactrue$ factor loading matrix,  $\Sigmaeps=\diagonalmatrix{\sigma_1^2,\ldots,\sigma_m^2}$ is diagonal and  $\nfactrue$ is the so-called factor dimension.
As discussed by the comprehensive textbooks of \cite{gor:fac} and \cite{and:int}, the multi-factor model is interesting both from a mathematical and a statistical perspective.

Starting with the pioneering work of  \cite{rei:ide} and \cite{and-rub:sta}, the mathematical analysis centers around the question of identifiability of the parameter $\facloadtrue$  and $ \Sigmaeps$ for a given $\Omega$, both in situations where the assumed factor dimension is equal to or different from the true factor dimension $\nfactrue$, see e.g. \cite{sha:ide} and \cite{bek-ten:gen}. Many mathematical conditions have been proposed  to address the various types of unidentifiability inherent in any factor model, see \cite{Fruehwirth2023When} for a recent review.

One such condition is variance identification which ensures that the decomposition of $\Omega$ in (\ref{Omegatrue}) is unique in the following sense. For any pair $(\facloadmatrix, \Sigmak{\nfactrue})$,
where $\facloadmatrix$ is an $\dimy \times \nfactrue$ factor loading matrix  and $\Sigmak{\nfactrue}$ is a  diagonal matrix such that
$\Omega=\facloadmatrix \facloadmatrix  ^\top + \Sigmak{\nfactrue}$,
it follows that  $\Sigmak{\nfactrue}=\Sigmaeps$ and hence
the cross-covariance matrices $\facloadmatrix \facloadmatrix  ^\top =
\facloadtrue \facloadtrue ^\top$ are identical. In this case,  the underlying loading matrix can be identified up to rotational invariance (\cite{and-rub:sta}),  i.e. $\facloadmatrix = \facloadtrue P$  where $P$ is a permutation matrix. \cite{and-rub:sta} provide following sufficient condition for variance identification, also known as row deletion property: after deleting any row from $\facloadtrue$, the remaining matrix contains two disjoint submatrices of rank $\nfactrue$.
In the present paper, we contribute to the mathematical aspects of multi-factor analysis by proving a sufficient condition for the row-deletion property based on the zero-nonzero pattern of the factor loading matrix and show how it can be verified in practice through an efficient algorithm.

Variance identification  becomes vital when the factor dimension is unknown.  A typical example where variance identification is violated are so-called spurious factors, where only a single non-zero factor loading is present in the corresponding column of $\facloadtrue$. Such spurious factors emerge in particular, when a factor model of dimension $\nfac> \nfactrue$ is employed to explain the covariance matrix $\Omega$ emerging from model \eqref{Omegatrue}.
\cite{rei:ide} shows that if there is a solution to the decomposition
\eqref{Omegatrue} with  factor
dimension $\nfactrue$, then there exist infinitely many solutions with a larger
factor dimension $\nfac >\nfactrue$.
\cite{tum-sat:ide} give a representation of these solutions, characterized by the $\dimy \times \nfac$ loading matrix $\facloadmatrix$ and the diagonal matrix  $\Sigmak{\nfac}$. They  show  that the true loading matrix $\facloadtrue$ is embedded within $\facloadmatrix$,
    but disguised  by spurious factors and a rotation $P$. For $\nfac= \nfactrue+1$, for instance,
    \begin{eqnarray} \label{overbate}
     \facloadmatrix =
     \left(\begin{array}{cc}
      \facloadtrue &  \Msp
     \end{array} \right) P ,  \quad
     \Sigmak{\nfac} = \Sigmaeps - \Msp \Msp ^\top
    \end{eqnarray}
   where $ \Msp $ is a spurious column with a single non-zero factor loading.
Obviously, the pair $(\facloadmatrix, \Sigmak{\nfac})$  implies the same  covariance matrix $\Omega$ as the true model and yields the same predictive distribution for $Y$ as the pair
 $(\facloadtrue, \Sigmaeps)$. On the other hand,
 model $(\facloadmatrix, \Sigmak{\nfac})$  features a factor loading matrix of dimension  $\nfac> \nfactrue$ and overestimates (inflates) the true factor dimension.
 However, a quick check immediately reveals that
 the $\dimy \times \nfac$ loading matrix  $\facloadmatrix$ violates variance identification, even if $\facloadtrue$ satisfies the conditions for variance identification: once the row corresponding to the single non-zero element in $M$ is removed, the remaining matrix has rank  $\nfactrue $ and does not contain two distinct sub-matrices of rank $\nfac > \nfactrue$.
 Consequently, any pair $(\facloadmatrix, \Sigmak{\nfac})$  that violates variance identification should not be considered a reliable representation for recovering the number of factors, as it very likely
overestimates the factor dimension.

This example clearly indicates that variance identification
is also relevant for statistical factor analysis, in particular  when the  factor dimension is unknown.
 Statistical analysis for factor models centers around fitting a suitable model to  realizations  $y_1, \ldots,y_T$  of $Y$, typically using ML estimation (\cite{and:bay,liu-rub:max,rub-tha:em_alg}) or a Bayesian inference. The Bayesian approach
 combines the likelihood derived from the factor model \eqref{Omegatrue}
 with a prior distribution on $\facloadtrue$ and $\Sigmaeps$ and
 offers several
 attractive features. The use of proper priors on the  idiosyncratic variances  $\sigma^2 _i$, for instance, avoids  Heywood problems  common in ML estimation, where some of the
 estimated $\sigma^2 _i$ are negative, see e.g. \cite{Fruehwirth2024Sparse}.

  ML and Bayesian approaches differ fundamentally when the factor dimension $\nfactrue$ is unknown. To choose $\nfactrue$, ML estimation employs an
incremental approach where a factor model with increasing factor dimension
is refitted to the data and BIC-type model selection criteria are applied to estimate $\nfactrue$, see e.g.
\cite{bai-ng:det2002}.
  In recent years, sparse Bayesian factor analysis became extremely popular in dealing with uncertainty regarding the factor dimension, see among many others \cite{luc-etal:spa,wes:bay_fac,gha-etal:bay,bha-dun:spa,con-etal:bay_exp,roc-geo:fas,kau-sch:bay,zha-etal:bay_gro,Fruehwirth2024Sparse}.
Sparse Bayesian factor analysis recovers the number of factors $\nfactrue$
from the data in a one-sweep algorithm. It
combines an overfitting factor model where the factor dimension is potentially bigger than $\nfactrue$ with a prior on the
 factor loading matrix
that introduces prior column sparsity in $\facloadtrue$.
This allows us to learn the number of factors on the fly and the resulting posterior distribution $p(\nfactrue\mid y_1, \ldots,y_T)$ allows uncertainty quantification with respect to $\nfactrue$.

\begin{figure}[t]
  \centering
  \includegraphics[width=0.4\textwidth,]{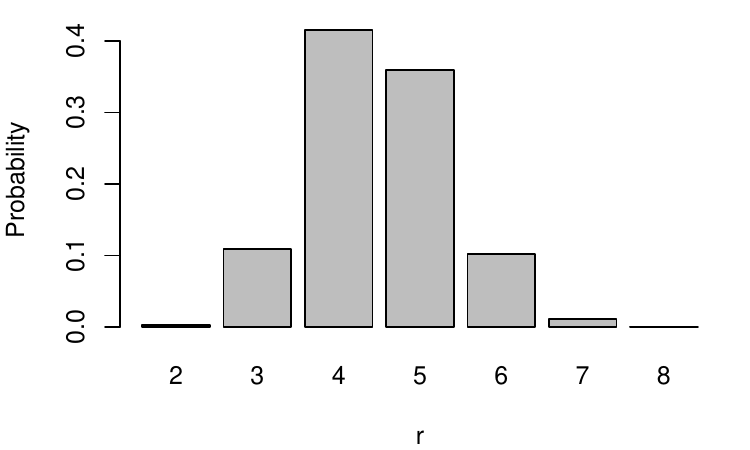}
  \includegraphics[width=0.4\textwidth,]{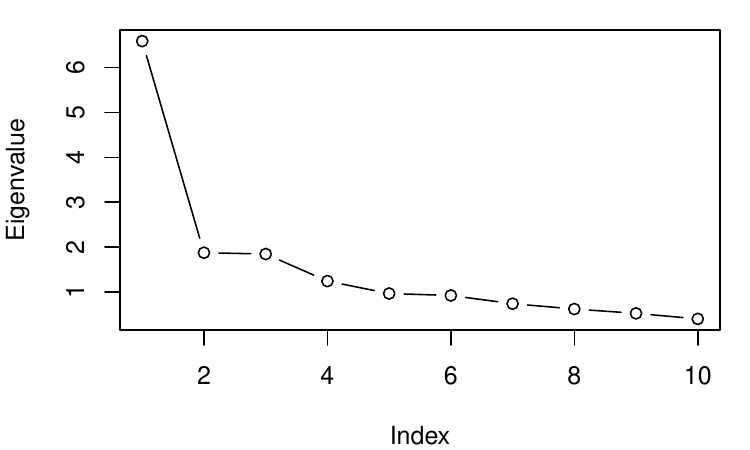}
  \caption{Posterior distribution $p(\nfactrue\mid y_1, \ldots,y_T)$ of the number factors for 52 weekly returns of 17 currencies of big trading partners of the Eurozone between January and December 2005 in comparison to the scree plot (right hand side).}
	\label{FIG1NEW}
\end{figure}

To illustrate one major difference between sparse Bayesian factor analysis and PCA,  Figure~\ref{FIG1NEW}  compares
the posterior distribution $p(\nfactrue\mid y_1, \ldots,y_T)$ of the number factors for the first data set considered in Section~\ref{appEx}, with the scree plot obtained from PCA.
The data are 52 weekly returns of 17 currencies of big trading partners of the Eurozone between January and December 2005.
As common for financial data, a strong market factor is present in the scree plot (shown on the right hand side) in addition to several weaker factors  which explain a small fraction of variance.
The posterior distribution (shown on the left hand side) translates the ambiguity in the scree plot into a posterior  distribution $p(\nfactrue\mid y_1, \ldots,y_T)$ that puts considerable mass on the presence of 4 or 5 factors.

The rest of the paper is organized as follows. Our mathematical results are summarized in Section~\ref{sec:modeltheory}
and applied to sparse Bayesian factor analysis in Section~\ref{sec:illustration}.
 While our motivation comes from sparse Bayesian factor analysis, our mathematical insights
are of a purely structural nature and potentially useful beyond this specific application.
After a brief review of variance identification in Section~\ref{sec:ident}, we provide
a counting rule on the zero-nonzero pattern
of the factor loading matrix
$\facloadtrue$, summarized in the binary matrix $\delta$, to check variance identification. This counting rule
was introduced by \cite{sat:stu} and
 only shown to be a necessary condition which has to be checked for $\facloadtrue$ and all possible rotations $\facloadmatrix =\facloadtrue P$. Recently,
\cite{Fruehwirth2023When} were able to prove that this counting rule is sufficient for variance identification, provided that $\facloadtrue$ exhibits a so-called  generalized lower triangular (GLT) structure. However, their proof heavily relies on assuming a GLT structure and is not easily extended to alternative structures or unconstrained loading matrices.
As a first major contribution, we prove in
Theorem~\ref{thm:suf} in Section~\ref{sec:sufficient} that this counting rule is sufficient for the row deletion property of \cite{and-rub:sta} except for a set of Lebesque measure zero. As opposed to  \cite{Fruehwirth2023When}, Theorem~\ref{thm:suf} does not require any structural constraints and can be applied to constrained and  unconstrained loading matrices alike.
Our proof relies on matching the binary zero-nonzero pattern $\delta$
of the factor loading matrix to a bipartite graph which is a mathematical object that captures the structure of $\delta$ but is invariant to permutations of the rows and columns of $\delta$, similar to the counting rule being invariant to permutations of the rows and columns of $\delta$.

Mathematically, the counting rule is a condition on all non-empty submatrices of
$\delta$ with $q \in \{ 1, \ldots, \nfactrue\}$ columns requiring that this submatrix has at least $2q+1$ non-zero rows. This condition could be checked for all $2^\nfactrue -1 $ submatrices which becomes infeasible for increasing $\nfactrue$. As a second major contribution of this paper, we design in Section~\ref{sec:fast} an efficient algorithm for checking the counting rule and prove  in Theorem~\ref{thm:ver} that this condition can be verified in polynomial time.
This algorithm is available as open source code\footnote{The source code is available at \url{https://hdarjus.github.io/sparvaride/}.} and can be applied regardless whether the loading matrix is constrained as in \cite{Fruehwirth2024Sparse} or unconstrained as in \cite{kau-sch:bay}.

In Section~\ref{sec:illustration}, we return to sparse Bayesian factor analysis. Posterior inference in sparse Bayesian factor analysis is typically performed using simulation techniques such as Markov chain Monte Carlo, see \cite{Fruehwirth2024Sparse} among many others. During sampling from the joint posterior distribution of all unknowns, identifiability conditions are either partially or completely ignored. Estimates of the quantities of interest are obtained by post-processing the posterior draws from such an unidentified model.
In Section~\ref{sec:illustration},
we investigate specifically the impact of conditions ensuring variance identification on recovering both the covariance matrix $\Omega$ (which is essential for prediction) as well as the true number of factors during such a post-processing step. We illustrate for simulated data that using only those posterior draws during post-processing  for which a sufficient condition for variance identification based on the counting rule of \cite{sat:stu} is fulfilled is instrumental in recovering the unknown factor dimension. On the other hand, recovering $\Omega$ and, hence, prediction is fairly robust to ignoring this condition. In addition, we estimate the number of factors in a financial application dynamically using a moving window approach over 100 overlapping periods. Whereas  prediction is again robust to the presence of posterior draws that are not variance identified, including this condition
avoids overfitting solutions that inflate the number of factors and
leads to a sparser number of factors explaining the observed variation of the data.
 We conclude the paper
with discussions included in Section~\ref{sec:conclusion}.

\section{Model and Theoretical Results} \label{sec:modeltheory}

\subsection{Variance Identification in the Basic Factor Model} \label{sec:ident}

Let $\yv=(y_1,\ldots,y_T)$ be a sequence of $\dimy$-dimensional observations, which are \SFS{centered around zero and} assumed to arise from a latent linear factor model with $\nfactrue$ factors,
\begin{equation}\label{eq:model}
  y_t = \facloadmatrix f_t + \epsilon_t,
\end{equation}
where $f_t$ is the $\nfactrue$-vector of latent factors, $\facloadmatrix$ is the $\dimy\times\nfactrue$-dimensional matrix of factor loadings $\facload{ij}$ with full column rank, and $\epsilon_t$ is the $\dimy$-vector of idiosyncratic errors, for $t=1,\ldots,T$.
In the basic factor model, the idiosyncratic errors are assumed to be iid $\dimy$-variate Gaussian \SFS{random variable} $\epsilon_t\sim N_\dimy(0,\Sigmaeps)$, where $\Sigmaeps=\diagonalmatrix{\sigma_1^2,\ldots,\sigma_m^2}$ is diagonal.
Furthermore, the latent factors are iid $\nfactrue$-variate Gaussian \SFS{random variable} $f_t\sim N_\nfactrue(0,I_\nfactrue)$, where $I_\nfactrue$ is the $\nfactrue\times\nfactrue$-dimensional identity matrix, and the factors are independent from the idiosyncratic errors $\epsilon_s$ for all $s=1,\ldots,T$.
When the latent factors are integrated out, this specification gives rise to the matrix decomposition of the covariance matrix $\Variance{y_t}=\Omega=\facloadmatrix\facloadmatrix^\top+\Sigmaeps$.

It is well known that any unitary matrix $G$ can be used to rotate the factor loadings $\facloadmatrix$ into $\facloadmatrix G$ without changing the covariance matrix $\Omega$.
In this case, model~\eqref{eq:model} changes to the observationally equivalent $y_t=(\facloadmatrix G)(G^{-1} f_t)+\epsilon_t$, and therefore $\facloadmatrix$ is not uniquely identified, which is called rotational invariance.
Further restrictions
are \SFS{required}
to achieve unique identification of $\facloadmatrix$,
\SFS{however we work with unconstrained loading matrices}
in this paper.

\SFS{As mentioned previously, we} contribute to the literature on the identification of $\Sigmaeps$, called variance identification.
More precisely, we consider the basic factor model as part of a sparse Bayesian factor analysis (BFA) model, where the factor loading matrix $\facloadmatrix$ follows an unknown zero-nonzero pattern that is estimated from the data along with the other parameters.
\SFS{Sparsity is allowed a priori, but not enforced, because the estimated pattern may be a fully nonzero matrix}.
In such a setting, the identifiability of $\Sigmaeps$ is not known a priori, because it depends on the estimated zero-nonzero pattern of $\facloadmatrix$.
Further details on the sparse BFA model are given in Section~\ref{sec:illustration}.

We emphasize, that the sufficient condition that we employ to achieve variance identification is a condition on $\facloadmatrix$.
In other words, we investigate the properties of $\facloadmatrix$ to say something about the identifiability of $\Sigmaeps$.
\cite{and-rub:sta} provide such a condition, later modified in~\cite{tum-sat:ide} for overfitting factor models, called the extended row deletion property in~\cite{Fruehwirth2023When}.

\begin{defn}[Extended Row Deletion Property, $\RDrs$] %
  The $\dimy\times\nfactrue$-dimensional factor loading matrix $\facloadmatrix$ is said to satisfy the extended row deletion property $\RDrs$ if for any $s>0$ rows of $\facloadmatrix$ that are removed, the remaining rows of $\facloadmatrix$ can be grouped into two matrices of rank $r$.
\end{defn}

\cite{and-rub:sta} show that $\RDr$ is a sufficient condition for the identification of $\Sigmaeps$.
Their result holds for every $\facloadmatrix$ with real-valued entries and thus also for a matrix with some exact zero entries; this makes it relevant for sparse BFA.
A second application of $\RDrs$ is variance identification in overfitting factor models, when series-specific \SFS{(spurious)} factors are allowed, which do not contribute to the off-diagonal elements of $\Omega$, as introduced by~\cite{tum-sat:ide}.
In both~\cite{and-rub:sta} and~\cite{tum-sat:ide}, however, the theory lacks practical ways to verify $\RDrs$.
Later, \cite{sat:stu} develop a new condition for $\facloadmatrix$, based on counting nonzero rows in all rotations of $\facloadmatrix$, which is shown to be necessary for $\RDrs$.
However, since the authors consider infinitely many rotations of $\facloadmatrix$, \SFS{this} %
condition  is %
unverifiable \SFS{in practice} and is not widely applied in the literature to our knowledge.

Recently, \cite{Fruehwirth2023When} directly build on results by~\cite{and-rub:sta},~\cite{tum-sat:ide}, and~\cite{sat:stu}, and introduce a framework in sparse BFA for the joint identification of $\facloadmatrix$ and $\Sigmaeps$.
The framework is based on an identifying assumption\footnote{%
$\facloadmatrix$ is assumed to adhere to the so-called generalized lower triangular structure.
For details, refer to~\cite{Fruehwirth2023When}.} %
on $\facloadmatrix$ combined with the following counting rule, which is in the same spirit as the counting rule of~\cite{sat:stu}.
\begin{defn}[Counting Rule, $\CRrs$]
  The $\dimy\times\nfactrue$-dimensional binary matrix $\binarymatrix$ is said to satisfy the counting rule $\CRrs$ if every submatrix containing $q$ columns of $\binarymatrix$ has at least $2q+s$ nonzero rows for every $1\le q\le r$.
\end{defn}
The authors show, among others, that, in their framework, an $\dimy\times\nfactrue$ binary matrix $\binarymatrix$ satisfying $\CRrs$ is sufficient for almost all $\dimy\times\nfactrue$ factor loading matrices $\facloadmatrix$ to satisfy $\RDrs$, where $\binarymatrix$ is an indicator matrix for $\facloadmatrix\neq0$; that is, $\facload{ij}=0$ if $\binarymatrix_{ij}=0$.
We generalize this result
\SFS{to unconstrained loading matrices} in the next section.

\subsection{Sufficient Condition for Identification} \label{sec:sufficient}

We fix the notation and the terminology for the rest of the section.
Denote by $\facloadmatrix$ an $\dimy\times\nfactrue$-dimensional matrix with real-valued elements $\facload{ij}$, where we pay attention to zeros for their special interpretation in sparse Bayesian factor analysis. %
Additionally, denote by $\binarymatrix$ an $\dimy\times\nfactrue$-dimensional binary matrix of zeros and ones.
We say that $\facloadmatrix$ is generated by $\binarymatrix$ if $\facload{ij}=0$ whenever $\binarymatrix_{ij}=0$, %
and we denote the set of all $\facloadmatrix$ generated by $\binarymatrix$ by $\sparsespace$.
We think of $\sparsespace$ as being equivalent to a continuous probability space on \SFS{$\mathbb{R}^d$, where $d={\sum_{i,j}\binarymatrix_{ij}}$ is equal to the total number of non-zero elements in $\facloadmatrix$.
This definition} is motivated by the slab elements of spike-and-slab priors on $\facloadmatrix$ in Bayesian variable selection (\cite{tadesse_handbook_2021}).
The probability space brings us to an expression that we extensively use below: ``for all $\facloadmatrix$ generated by $\binarymatrix$, except for a set of Lebesgue-measure zero'' is formally understood as ``the set of exceptional $\facloadmatrix$ matrices form a Lebesgue-nullset in $\sparsespace$''.

This section is mainly concerned with \SFS{proving} the sufficiency of $\CRrs$ for $\RDrs$, formally stated in %
Theorem~\ref{thm:suf} \SFS{at the end of this section.
The} proof is presented in multiple pieces through Propositions~\ref{lem:s0}, \ref{lem:red5}, and \ref{lem:cr2}, and surrounding lemmas.
Since both $\CRrs$ and $\RDrs$ imply $\dimy\ge2\nfactrue+s$ for $\dimy\times\nfactrue$ matrices, we assume that $\dimy\ge2\nfactrue+s$ throughout this section.
First,
\SFS{we show in Proposition~\ref{lem:s0} that
the problem can be} reduced to the case of $s=0$.
Later, 
\SFS{based on Proposition~\ref{lem:s0},} we only need to prove that $\CRro$ implies $\RDro$, except for a set of Lebesgue-measure zero.

\begin{prop}\label{lem:s0}
  Let $\binarymatrix$ denote an $\dimy\times\nfactrue$-dimensional binary matrix and $s\ge0$ a nonnegative integer.
  Consider the following statement: if $\binarymatrix$ satisfies the counting rule $\CRrs$, then the row deletion property $\RDrs$ holds for all $\facloadmatrix$ generated by $\binarymatrix$, except for a set of Lebesgue-measure zero.
  If
  that statement holds for $s=0$, then
  it also holds for %
  \SFS{all positive integers $0< s \le
  \dimy - 2\nfactrue$. }
\end{prop}
\begin{proof}[Proof of Proposition~\ref{lem:s0}]
   Let $\binarymatrix$ be given such that it satisfies $\CRrs$, and denote by $\sparsespace$ the set of all factor loading matrices $\facloadmatrix$ generated by $\binarymatrix$.
   We want to show that $\RDrs$ holds for all $\facloadmatrix\in\sparsespace$ except for a set of Lebesgue-measure zero, or, equivalently, that after deleting any $s$ rows from these factor loading matrices, their remaining rows can almost surely be grouped into two matrices of rank $\nfactrue$.
    Let us %
 \SFS{choose an arbitrary set of $s$ rows.}
  After deleting these $s$ rows, $\binarymatrix$ becomes $\binarymatrixrem$,
  \SFS{containing the remaining rows.}
  The resulting set of factor loading matrices generated by $\binarymatrixrem$ is denoted by $\sparsespacerem$, and
  \SFS{the elements of $\sparsespacerem$
  are denoted by $\facloadmatrixrem$.}
  It is easy to see that $\binarymatrixrem$ satisfies $\CRro$.
   \SFS{If the statement holds for $s=0$, then
    $\RDro$ is satisfied for all  $\facloadmatrixrem$ generated by $\binarymatrixrem$, except for a set of Lebesgue-measure zero.}
   Therefore, by assumption, the set $\mathcal{N}$ of factor loading matrices $\facloadmatrixrem$ that do not satisfy $\RDro$ is of Lebesgue-measure zero in $\sparsespacerem$.
  Now, we map back from $\sparsespacerem$ to $\sparsespace$ instead of working with specific matrices $\facloadmatrixrem$ and $\facloadmatrix$, and that is how we avoid taking the union of uncountably many Lebesgue-nullsets.
  Observe that $\sparsespace$ is the product space of $\sparsespacerem$ and the space spanned by the $s$ deleted rows, and the measure on $\sparsespace$ is the product of the Lebesgue-measures on its two said subspaces.
  Therefore, the set of factor loading matrices $\facloadmatrix$ that correspond to $\mathcal{N}$, i.e., that do not satisfy $\RDro$ after deleting the $s$ test rows, is of Lebesgue-measure zero in $\sparsespace$.
  Equivalently, the set of factor loading matrices $\facloadmatrix$, whose rows can be grouped into two matrices of rank $\nfactrue$ after deleting the fixed $s$ test rows, is of Lebesgue-measure one in $\sparsespace$.

  The above argument holds for any $s$ test rows, and there exist only finitely many ways to choose $s$
  \SFS{among the $\dimy$ rows of} %
  $\binarymatrix$.
  The set of factor loading matrices $\facloadmatrix$ generated by $\binarymatrix$ that satisfy $\RDrs$ is the intersection of all the sets that result from different choices of $s$ test rows to delete.
  Finite intersection of sets of Lebesgue-measure one is of Lebesgue-measure one, and the proof is thus complete.
\end{proof}

\SFS{Before we exploit the simplification provided
by Proposition~\ref {lem:s0} below in Propositions~\ref{lem:red5} and %
\ref{lem:cr2}},
we take a \SFS{detour} %
into elementary graph theory to prove our results
in a structured way.
The centerpiece of the \SFS{final} proof of Theorem~\ref{thm:suf} is the classical duality theorem by K\H{o}nig~\cite{kon:gra} and Egerv\'ary~\cite{ege:mat} in graph theory, which is stated later in this section, where we put $\CRro$ and $\RDro$ on the two sides of the duality.
Graph theory also provides us with a convenient representation of the problem through a specific mapping of a binary matrix $\binarymatrix$ to its corresponding bipartite graph,
\SFS{which will be}
introduced in Definition~\ref{def:bin}.
Notably, these bipartite graphs are an equivalent representation of $\binarymatrix$ up to reordering of its rows and columns.
This is a helpful framework, since both $\CRro$ and $\RDro$ are invariant to row and column permutations.
In the following, we think of a nonzero entry $\binarymatrix_{ij}=1$ in $\binarymatrix$ as a line connecting a point $z_i$ that represents row $i$ with another point $c_j$ that represents column $j$.
This construction is formally defined next in order to provide the language for our proofs.
For a more detailed introduction to graph theoretic notions, see chapters 3.1 and 3.2 of \cite{wes:gra}.

\begin{defn}[Bipartite Graph and Bi-adjacency Matrix] \label{def:bin}
  A bipartite graph $B=(\verrowvertices,\vercolvertices,E_B)$ is a triplet of two disjoint sets of points, called vertices, $\verrowvertices$ and $\vercolvertices$, and a set of undirected lines, called edges, $E_B\subseteq\{\{z_i,c_j\}:z_i\in\verrowvertices,c_j\in\vercolvertices\}$.
  Given an $\dimy\times\nfactrue$-dimensional binary matrix $\binarymatrix$, an equivalent representation, up to reordering rows and columns of $\binarymatrix$, is the following bipartite graph.
  Vertices in the sets $\verrowvertices=\{z_1,\ldots,z_\dimy\}$ and $\vercolvertices=\{c_1,\ldots,c_\nfactrue\}$ correspond to the rows and, respectively, columns of $\binarymatrix$.
  An edge $e\in E_B$ is drawn between $z_i\in\verrowvertices$ and $c_j\in\vercolvertices$ if the corresponding matrix element $\binarymatrix_{ij}$ is nonzero; $z_i$ and $c_j$ are called the endpoints of $e$.
  Then, $\binarymatrix$ describes which pairs of vertices of $B$ are adjacent (i.e., connected by an edge) and therefore $\binarymatrix$ is called the bi-adjacency matrix of $B$.
  Turned around, we call $B$ the bipartite graph of $\binarymatrix$.
\end{defn}

Figure~\ref{fig:csm} shows a \SFS{$3 \times 3$} binary matrix $\binarymatrix$ and its bipartite graph $B=(\verrowvertices,\vercolvertices,E_B)$ with $\verrowvertices=\{z_1,z_2,z_3\}$ and $\vercolvertices=\{c_1,c_2,c_3\}$.
The edges are $E_B=\{\{z_1,c_1\},\{z_1,c_2\},\{z_1,c_3\},\{z_2,c_1\},\{z_3,c_1\}\}$.
One way we use graph theory is %
\SFS{displayed in the same figure:}
thick edges correspond to a so-called maximal matching in $B$, formally introduced below in Definition~\ref{def:matching}.
\SFS{Loosely speaking, we are trying to reorganize the rows of $\binarymatrix$ such that the main diagonal exhibits as many ones as possible. }
\SFS{For the  $\binarymatrix$ considered in Figure~\ref{fig:csm},} swapping rows $z_1$ and $z_2$ results in a reorganized $\binarymatrix$ with a diagonal of two ones, shown in bold face in the original $\binarymatrix$, which is the \SFS{maximum number that can be achieved in this case}.
\SFS{However, if element $\binarymatrix_{3,3}$ were also a one, then
 $\binarymatrix$ could be reorganized such that the diagonal was full and contained ones, only. In this case,} the maximum matching in $B$ would be of size three.

Let us zoom out of this particular example: the existence of a
\SFS{submatrix with a full diagonal of ones}  will be the proof of full-rankedness of all factor loading matrices $\facloadmatrix$ generated by $\binarymatrix$, except for a set of Lebesgue-measure zero. \SFS{This is stated in Lemma~\ref{lem:red}.}
\SFS{To generalize Lemma~\ref{lem:red}, we will then introduce the notion of a  matching in bipartite graphs.
It will be shown that graph matching allows us to %
simplify the search for a submatrix  $\squarebinarymatrix$ that satisfies the conditions of Lemma~\ref{lem:red}. Instead of %
rearranging the rows of $\binarymatrix$  till a suitable $\squarebinarymatrix$ is found,  a dual %
optimization task is performed on the bipartite graph $B$ corresponding to $\delta$.}

\begin{figure}[!t]
  \centering
  \begin{tikzpicture}
    \node [matrix] (delta) at (-4.5, 0) {
      \node (delta corner) {}; & \node [inner ysep=2pt, inner xsep=0pt] (c1) {$c_1$}; & \node {$c_2$}; & \node [inner ysep=2pt, inner xsep=0pt] (c3) {$c_3$}; \\
      \node [inner xsep=2pt, inner ysep=0pt] (r1) {$z_1$}; & \node{$1$}; & \node {$\textbf{1}$}; & \node{$1$}; \\
      \node {$z_2$}; & \node {$\textbf{1}$}; & \node {$0$}; & \node {$0$}; \\
      \node [inner xsep=2pt, inner ysep=0pt] (r3) {$z_3$}; & \node {$1$}; & \node {$0$}; & \node {$0$}; \\
    };
    \node [venn diagram, fit=(c1) (c3)] {};
    \node [venn diagram, rotate fit=90, fit=(r1) (r3)] {};

    \node (hiff) [right of=delta, xshift=1cm] {$\qquad\longmapsto\qquad$};
    \begin{scope}[right of=hiff, xshift=-2cm]
      \def\stepsize{0.3cm}
      \begin{scope}[xshift=0*\stepsize, yshift=3*\stepsize, start chain=going right, node distance=\stepsize]
        \node foreach \i in {1,2,3} [on chain, inner ysep=2pt, inner xsep=0pt] (fg\i) {$c_{\i}$};
      \end{scope}
      \begin{scope}[xshift=0*\stepsize, yshift=-3*\stepsize, start chain=going right, node distance=\stepsize]
        \node foreach \i in {1,2,3} [on chain, inner ysep=2pt, inner xsep=0pt] (sg\i) {$z_{\i}$};
      \end{scope}
    \end{scope}
    \node (fg) [venn diagram, fit=(fg1) (fg3)] {};
    \node (sg) [venn diagram, fit=(sg1) (sg3)] {};
    \begin{scope}[line width=0.5mm]
      \draw [edge] (fg1) -- (sg2);
      \draw [edge] (fg2) -- (sg1);
    \end{scope}
    \draw [edge] (fg1) -- (sg1);
    \draw [edge] (fg1) -- (sg3);
    \draw [edge] (fg3) -- (sg1);
  \end{tikzpicture}
  \caption[Example binary matrix and its bipartite graph]{$3\times3$ binary matrix $\binarymatrix$ (left) and its bipartite graph $B=(\verrowvertices,\vercolvertices,E_B)$ (right) with $\verrowvertices=\{z_1,z_2,z_3\}$ and $\vercolvertices=\{c_1,c_2,c_3\}$.
    One matching in $B$ is $\{\{z_1,c_2\},\{z_2,c_1\}\}$, which is displayed in bold face in $\binarymatrix$ and with thick edges in $B$.}
  \label{fig:csm}
\end{figure}

\begin{lem}\label{lem:red}
  Consider an $\nfactrue\times\nfactrue$-dimensional square submatrix $\squarebinarymatrix$ of an $\dimy\times\nfactrue$-dimensional binary matrix $\binarymatrix$.
  Assume that the diagonal of $\squarebinarymatrix$ only has ones and no zeros.
  Then, for all square factor loading submatrices $\squarefacloadmatrix$ generated by $\squarebinarymatrix$, $\squarefacloadmatrix$ is non-singular, except for a set of Lebesgue-measure zero.
\end{lem}
\begin{proof}[Proof of Lemma~\ref{lem:red}]
  The proof is based on the Leibniz formula for determinants; the well-known non-recursive formula that constructs the determinant of a square matrix as the sum of all products of $\nfactrue$ elements of the matrix, each taken from a different row and column and each multiplied by a sign that depends on the permutation of the columns.
  For the proof, we omit the square notation and use $\facloadmatrix$ and $\binarymatrix$ instead of $\squarefacloadmatrix$ and, resp., $\squarebinarymatrix$.

  The diagonal of $\facloadmatrix$ is a non-degenerate set of $\nfactrue$ continuous variables over $\Real{\nfactrue}$, all
  \SFS{of which} are almost surely nonzero by the definition of being generated by $\binarymatrix$, and their product constitutes one summand in the Leibniz formula (potentially after a sign switch).
  There may be further nonzero summands, but only finitely many, and they all include off-diagonal %
  \SFS{elements of} $\facloadmatrix$.
  Therefore, the sum is nonzero with probability one.
  Hence, $\det(\facloadmatrix)\neq0$ and $\facloadmatrix$ is non-singular with probability one.
\end{proof}

A matching in a bipartite graph implicitly encodes both a subset of rows and their permutation in the bi-adjacency matrix $\binarymatrix$
\SFS{in such a way}
that the permuted rows form a submatrix in $\binarymatrix$ with a full diagonal of ones.
The key is the unique correspondence between row indices and column indices that take part in the matching.
A small example %
\SFS{has been provided}
in Figure~\ref{fig:csm}, and larger examples are the thick edges in the two graphs in Figure~\ref{fig:corr}.
In the latter, for instance, by reordering vertices $z_1,\ldots,z_7$ such that thick edges do not cross, we obtain a full diagonal of ones in the reordered $\binarymatrix$.

\begin{defn}[Matching] \label{def:matching}
  A matching in a bipartite graph $B=(\verrowvertices,\vercolvertices,E_B)$ is a set of pairwise disconnected edges $M\subseteq E_B$ such that no two edges of $M$ have the same endpoints.
  A maximum matching in $B$ is a matching with the maximal number of edges among all matchings in $B$.\footnote{%
    Observe that there may exist matchings in $B$ that cannot be extended to larger matchings but are not maximum matchings, which complicates the search for a maximum.
    In Figure~\ref{fig:csm}, there are many matchings: e.g., $\emptyset$, $\{\{z_1,c_2\},\{z_2,c_1\}\}$, and $\{\{z_1,c_1\}\}$; the latter is not maximal but cannot be extended further.
    Edge set $\{\{z_1,c_1\},\{z_1,c_2\}\}$ is not a matching because $z_1$ appears twice as endpoint.
  }
  A $\vercolvertices$-saturating matching is a matching in $B$ that covers all vertices of $\vercolvertices$, i.e., it contains all vertices in $\vercolvertices$ as endpoints.
\end{defn}

The following statement generalizes Lemma~\ref{lem:red}
\SFS{where $\squarebinarymatrix$ was supposed to exhibit a
full diagonal of ones.}
In Lemma~\ref{lem:red3}, we implicitly find $\squarebinarymatrix$
even if the \SFS{rows of $\delta$ are not ordered such that an
approriate submatrix with a full diagonal of ones exists. We %
show that we can find $\squarebinarymatrix$  with the help of a $\vercolvertices$-saturating matching in the bipartite graph  $B$ of $\delta$}.

\begin{lem}\label{lem:red3}
  Consider an $\dimy\times\nfactrue$-dimensional binary matrix $\binarymatrix$ and its bipartite graph $B=(\verrowvertices,\vercolvertices,E_B)$.
  Assume that there is a $\vercolvertices$-saturating matching in $B$.
  Then, for all factor loading matrices $\facloadmatrix$ generated by $\binarymatrix$, there \SFS{exists} %
  a non-singular $\nfactrue\times\nfactrue$-dimensional submatrix $\squarefacloadmatrix$ in $\facloadmatrix$, except for a set of Lebesgue-measure zero.
\end{lem}

\begin{proof}[Proof of Lemma~\ref{lem:red3}]
  The $\vercolvertices$-saturating matching $M$ in $B$ consists of $\nfactrue$ edges $\{\{z_{i_1},c_1\},\ldots,\{z_{i_\nfactrue},c_\nfactrue\}\}$ with indices $1\le i_1,\ldots,i_\nfactrue\le\dimy$ and $i_j\neq i_l$ for $j\neq l$, which implies $\dimy\ge\nfactrue$.
  Then, the rows of $\binarymatrix$ can be reordered
  \SFS{such that the first $\nfactrue$ rows of $\binarymatrix$ are equal to}
 the $\nfactrue$ rows $z_{i_1},\ldots,z_{i_\nfactrue}$.
  Denote by \SFS{$\tilde{\binarymatrix}$} %
  the reordered $\binarymatrix$ and by $\squarebinarymatrix$ the top $\nfactrue$ rows of \SFS{$\tilde{\binarymatrix}$}. %
  The diagonal of $\squarebinarymatrix$ directly corresponds to the matching $M$ and therefore contains only ones.
  Due to Lemma~\ref{lem:red}, the top $\nfactrue$ rows of any factor loading matrix
  \SFS{$\tilde{\facloadmatrix}$} %
  generated by
  \SFS{$\tilde{\binarymatrix}$} %
  constitute a non-singular matrix, except for a set of Lebesgue-measure zero.
  Now, we can revert \SFS{$\tilde{\facloadmatrix}$} %
  to the original ordering and obtain a factor loading matrix $\facloadmatrix$ generated by the original binary matrix $\binarymatrix$.
  Since the rank is invariant to reordering of rows and $z_{i_1},\ldots,z_{i_\nfactrue}$ are also rows of $\facloadmatrix$, $\facloadmatrix$ almost surely contains a non-singular $\nfactrue\times\nfactrue$-dimensional submatrix.
\end{proof}

In the proof, we use a correspondence between a full-column-rank submatrix in $\facloadmatrix$, generated by $\binarymatrix$, and a $\vercolvertices$-saturating matching in the bipartite graph $B$ of $\binarymatrix$.
\SFS{For illustration, consider the matrix} $\binarymatrix$ and its bipartite graph $B=(\verrowvertices,\vercolvertices,E_B)$ in \SFS{the top of} Figure~\ref{fig:corr}, where $\verrowvertices=\{z_1,\ldots,z_7\}$, $\vercolvertices=\{c_1,c_2,c_3\}$, and $E_B$ is the set of edges connecting $\vercolvertices$ and $\verrowvertices$.
The $\vercolvertices$-saturating matching $\{\{z_1,c_2\},\{z_2,c_3\},\{z_4,c_1\}\}$ is shown using thick edges in $B$, and the same positions in $\binarymatrix$ are typed in bold face. %
These elements form a full diagonal of three ones in $\binarymatrix$ after reordering its rows beginning with $z_4$, $z_1$, and $z_2$.

\begin{figure}[!t]
  \centering
  \begin{tikzpicture}
    \node [matrix] (delta) at (-4.5, 0) {
      \node (delta corner) {}; & \node [inner xsep=0pt] (c1) {$c_1$}; & \node {$c_2$}; & \node [inner xsep=0pt] (c3) {$c_3$}; & \node {\phantom{$c_1^{*}$}}; & \node {\phantom{$c_2^{*}$}}; & \node [inner sep=0pt] (cg3) {\phantom{$c_3^{*}$}}; \\
      \node [inner ysep=0pt] (r1) {$z_1$}; & \node {1}; & \node {\textbf{1}}; & \node {0}; \\
      \node {$z_2$}; & \node {1}; & \node {0}; & \node {\textbf{1}}; \\
      \node {$z_3$}; & \node {0}; & \node {0}; & \node {1}; \\
      \node {$z_4$}; & \node {\textbf{1}}; & \node {0}; & \node {0}; \\
      \node {$z_5$}; & \node {1}; & \node {1}; & \node {0}; \\
      \node {$z_6$}; & \node {1}; & \node {0}; & \node {0}; \\
      \node [inner ysep=0pt] (r7) {$z_7$}; & \node (c1r7) {0}; & \node {0}; & \node (c3r7) {1}; \\
    };
    \node [venn diagram, fit=(c1) (c3)] {};
    \node [venn diagram, rotate fit=90, fit=(r1) (r7)] {};

    \node (hiff) [right of=delta, xshift=2.2cm] {$\qquad\longmapsto\qquad$};

    \begin{scope}[right of=hiff, xshift=-0.8cm]
      \def\stepsize{0.3cm}
      \begin{scope}[xshift=1.5*\stepsize, yshift=4*\stepsize, start chain=going right, node distance=\stepsize]
        \node foreach \i in {1,2} [on chain] (fg\i) {$c_{\i}$};
        \node [on chain,inner xsep=0.2em,inner ysep=0.2em] (fg3) {$c_{3}$};
        \node foreach \i in {1,2,3} [on chain] (fgg\i) {\phantom{$c_{\i}^{*}$}};
      \end{scope}
      \begin{scope}[xshift=0.5*\stepsize, yshift=-4*\stepsize, start chain=going right, node distance=\stepsize]
        \node foreach \i in {1,2,3,4,5,6,7} [on chain] (sg\i) {$z_{\i}$};
      \end{scope}
    \end{scope}

    \node [venn diagram, fit=(fg1) (fg3)] {};
    \node [venn diagram, fit=(sg1) (sg7)] {};

    \begin{scope}[line width=0.5mm]
      \draw [edge] (fg2) -- (sg1);
      \draw [edge] (fg3) -- (sg2);
      \draw [edge] (fg1) -- (sg4);
    \end{scope}

    \draw [edge] (fg3) -- (sg3);
    \draw [edge] (fg2) -- (sg5);
    \draw [edge] (fg1) -- (sg6);

    \draw [edge] (fg1) -- (sg1);
    \draw [edge] (fg1) -- (sg2);
    \draw [edge] (fg3) -- (sg3);
    \draw [edge] (fg1) -- (sg5);
    \draw [edge] (fg3) -- (sg7);

    \begin{scope}[decoration={brace,amplitude=10pt,mirror,raise=4pt}]
      \draw [decorate] (c1r7.south west) -- (c3r7.south east) node [black,midway,yshift=-0.7cm] {$\binarymatrix$};
      \draw [decorate] (sg1.south west) -- (sg7.south east) node [black,midway,yshift=-0.7cm] {$B$};
    \end{scope}

    \begin{scope}[yshift=-5cm]
    \node [matrix] (delta) at (-4.5, 0) {
      \node (delta corner) {}; & \node [inner xsep=0pt] (c1) {$c_1$}; & \node {$c_2$}; & \node {$c_3$}; & \node {$c_1^{*}$}; & \node {$c_2^{*}$}; & \node [inner sep=0pt] (cg3) {$c_3^{*}$}; \\
      \node [inner ysep=0pt] (r1) {$z_1$}; & \node {1}; & \node {\textbf{1}}; & \node {0}; & \node {1}; & \node {1}; & \node {0}; \\
      \node {$z_2$}; & \node {1}; & \node {0}; & \node {\textbf{1}}; & \node {1}; & \node {0}; & \node {1}; \\
      \node {$z_3$}; & \node {0}; & \node {0}; & \node {1}; & \node {0}; & \node {0}; & \node {\textbf{1}}; \\
      \node {$z_4$}; & \node {\textbf{1}}; & \node {0}; & \node {0}; & \node {1}; & \node {0}; & \node {0}; \\
      \node {$z_5$}; & \node {1}; & \node {1}; & \node {0}; & \node {1}; & \node {\textbf{1}}; & \node {0}; \\
      \node {$z_6$}; & \node {1}; & \node {0}; & \node {0}; & \node {\textbf{1}}; & \node {0}; & \node {0}; \\
      \node [inner ysep=0pt] (r7) {$z_7$}; & \node (c1r7) {0}; & \node {0}; & \node (c3r7) {1}; & \node (c1gr7) {0}; & \node {0}; & \node (c3gr7) {1}; \\
    };
    \node [venn diagram, fit=(c1) (cg3)] {};
    \node [venn diagram, rotate fit=90, fit=(r1) (r7)] {};

    \node (hiff) [right of=delta, xshift=2.2cm] {$\qquad\longmapsto\qquad$};

    \begin{scope}[right of=hiff, xshift=-0.8cm]
      \def\stepsize{0.3cm}
      \begin{scope}[xshift=1.5*\stepsize, yshift=4*\stepsize, start chain=going right, node distance=\stepsize]
        \node foreach \i in {1,2,3} [on chain] (fg\i) {$c_{\i}$};
        \node foreach \i in {1,2} [on chain] (fgg\i) {$c_{\i}^{*}$};
        \node [on chain,inner xsep=0.2em,inner ysep=0pt] (fgg3) {$c_{3}^{*}$};
      \end{scope}
      \begin{scope}[xshift=0.5*\stepsize, yshift=-4*\stepsize, start chain=going right, node distance=\stepsize]
        \node foreach \i in {1,2,3,4,5,6,7} [on chain] (sg\i) {$z_{\i}$};
      \end{scope}
    \end{scope}

    \node [venn diagram, fit=(fg1) (fgg3)] {};
    \node [venn diagram, fit=(sg1) (sg7)] {};

    \begin{scope}[line width=0.5mm]
      \draw [edge] (fg2) -- (sg1);
      \draw [edge] (fg3) -- (sg2);
      \draw [edge] (fg1) -- (sg4);

      \draw [edge] (fgg3) -- (sg3);
      \draw [edge] (fgg2) -- (sg5);
      \draw [edge] (fgg1) -- (sg6);
    \end{scope}

    \draw [edge] (fg3) -- (sg3);
    \draw [edge] (fg2) -- (sg5);
    \draw [edge] (fg1) -- (sg6);

    \draw [edge] (fg1) -- (sg1);
    \draw [edge] (fg1) -- (sg2);
    \draw [edge] (fg3) -- (sg3);
    \draw [edge] (fg1) -- (sg5);
    \draw [edge] (fg3) -- (sg7);

    \draw [edge] (fgg2) -- (sg1);
    \draw [edge] (fgg3) -- (sg2);
    \draw [edge] (fgg1) -- (sg4);

    \draw [edge] (fgg1) -- (sg1);
    \draw [edge] (fgg1) -- (sg2);
    \draw [edge] (fgg3) -- (sg3);
    \draw [edge] (fgg1) -- (sg5);
    \draw [edge] (fgg3) -- (sg7);

    \begin{scope}[decoration={brace,amplitude=10pt,mirror,raise=4pt}]
      \draw [decorate] (c1r7.south west) -- (c3r7.south east) node [black,midway,yshift=-0.7cm] {$\binarymatrix$};
      \draw [decorate] (c1gr7.south west) -- (c3gr7.south east) node [black,midway,yshift=-0.7cm] {copy of $\binarymatrix$};
      \draw [decorate] ([yshift=-0.9cm] c1r7.south west) -- ([yshift=-0.9cm] c3gr7.south east) node [black,midway,yshift=-0.7cm] {$\binarymatrix^{||}$};
      \draw [decorate] (sg1.south west) -- (sg7.south east) node [black,midway,yshift=-0.7cm] {$B^{||}$};
    \end{scope}
    \end{scope}
  \end{tikzpicture}
  \caption{A $7\times3$-dimensional binary matrix $\binarymatrix$ (top left), its bipartite graph $B$ (top right), DB matrix $\binarymatrix^{||}$ (bottom left) and DB graph $B^{||}$ (bottom right).
  Thick edges in $B$ (respectively, $B^{||}$) correspond to a maximal matching in $B$ ($B^{||}$), which is a $\vercolvertices$-saturating matching in $B$ ($\vercolvertices^{||}$-saturating matching in $B^{||}$).}
  \label{fig:corr}
\end{figure}

\SFS{So far, we} have discussed a condition that ensures the almost sure existence of a full-column-rank submatrix in $\facloadmatrix$ generated by $\binarymatrix$.
To satisfy $\RDro$, however, we look for two
\SFS{distinct}
submatrices of $\facloadmatrix$ that are of rank $\nfactrue$, i.e. two disjoint sets of $\nfactrue$ rows $(\facloadmatrix_{i_1,\cdot}, \ldots, \facloadmatrix_{i_{\nfactrue},\cdot})$ and $(\facloadmatrix_{i_{\nfactrue+1},\cdot}, \ldots, \facloadmatrix_{i_{2\nfactrue},\cdot})$ in $\facloadmatrix$ such that the corresponding submatrices $(\facloadmatrix_{i_1,\cdot}^\top \ldots \facloadmatrix_{i_{\nfactrue},\cdot}^\top)^\top$
and $(\facloadmatrix_{i_{\nfactrue+1},\cdot}^\top \ldots \facloadmatrix_{i_{2\nfactrue},\cdot}^\top)^\top$ are both of rank $\nfactrue$.
According to Lemma~\ref{lem:red}, assuming for now that the rows of $\facloadmatrix$ are ordered
\SFS{appropriately}, %
this amounts to finding two disjoint $\nfactrue\times\nfactrue$-dimensional submatrices of $\binarymatrix$ with a diagonal of ones, up to a Lebesgue-nullset.
In Lemma~\ref{lem:red2}, we replace this search task with a simpler one
\SFS{by introducing the notion of {\em duplicated binary (DB) matrices}.} An example \SFS{of a binary matrix and the corresponding DB matrix}
is
\SFS{provided in}
Figure~\ref{fig:corr} for $\dimy=7$ and $\nfactrue=3$.

\begin{defn}[DB Matrix $\binarymatrix^{||}$ of $\binarymatrix$]
  The Duplicated Binary (DB) matrix $\binarymatrix^{||}$ of $\dimy\times\nfactrue$-dimensional binary matrix $\binarymatrix$ is the $\dimy\times(2\nfactrue)$-dimensional binary matrix $(\binarymatrix\ \binarymatrix)$.
\end{defn}

\begin{lem}\label{lem:red2}
  Consider an $\dimy\times\nfactrue$-dimensional binary matrix $\binarymatrix$
  \SFS{with $\dimy\ge2\nfactrue$} and its DB matrix $\binarymatrix^{||}$.
  Then, the following two tasks are equivalent in the sense that the same row-indices $i_1,\ldots,i_\nfactrue$, $i_{\nfactrue+1},\ldots,i_{2\nfactrue}$, if any, solve both:
  \begin{enumerate}[label=(\roman*)]
    \item Find two disjoint $\nfactrue\times\nfactrue$-dimensional submatrices of $\binarymatrix$ with a diagonal of ones.
    \item Find a $(2\nfactrue)\times(2\nfactrue)$-dimensional submatrix of $\binarymatrix^{||}$ with a diagonal of ones.
  \end{enumerate}
\end{lem}
\begin{proof}[Proof of Lemma~\ref{lem:red2}]
  The key observation is that we are not concerned about the off-diagonal elements in this lemma.
  If row indices $i_1,\ldots,i_\nfactrue$, and $i_{\nfactrue+1},\ldots,i_{2\nfactrue}$ solve (i), then matrix
  \begin{equation} \label{eq:bin2}
    \begin{pmatrix}
      \binarymatrix_{i_1,\cdot} & \binarymatrix_{i_1,\cdot} \\
      \vdots & \vdots \\
      \binarymatrix_{i_\nfactrue,\cdot} & \binarymatrix_{i_\nfactrue,\cdot} \\
      \binarymatrix_{i_{\nfactrue+1},\cdot} & \binarymatrix_{i_{\nfactrue+1},\cdot} \\
      \vdots & \vdots \\
      \binarymatrix_{i_{2\nfactrue},\cdot} & \binarymatrix_{i_{2\nfactrue},\cdot}
    \end{pmatrix}
  \end{equation}
  constitutes rows $i_1,\ldots,i_{2\nfactrue}$ of $\binarymatrix^{||}$, which has a diagonal of ones.
  Thus, $i_1,\ldots,i_{2\nfactrue}$ solve (ii).

  If row indices $i_1,\ldots,i_{2\nfactrue}$ solve (ii), then the upper left $\nfactrue\times\nfactrue$-dimensional submatrix of $\binarymatrix^{||}$ is a submatrix of $\binarymatrix$ with a diagonal of ones, as seen in matrix~\eqref{eq:bin2}; the same holds for the lower right $\nfactrue\times\nfactrue$-dimensional submatrix.
  Thus, $i_1,\ldots,i_\nfactrue$, and $i_{\nfactrue+1},\ldots,i_{2\nfactrue}$ solve (i).
\end{proof}

\SFS{To utilize Lemma~\ref{lem:red2} for our final goal,
we resort to graph theory to verify
the existence of a $(2\nfactrue)\times(2\nfactrue)$-dimensional submatrix with a diagonal ones in the DB matrix $\binarymatrix^{||}$ of $\delta$.
To this aim we introduce  the notation of  a duplicated bipartite (DB) graph   of a binary matrix $\delta$ in Definition~\ref{def:DPGraph}.}
\SFS{For illustration, the} DB matrix $\binarymatrix^{||}$ and
\SFS{the corresponding} DB graph, denoted by $B^{||}$,
are shown in Figure~\ref{fig:corr} for a $7\times3$-dimensional binary matrix $\binarymatrix$.
\SFS{Lemma~\ref{lem:red4} then relates the existence of a  $\vercolvertices^{||}$-saturating matching in $B^{||}$ to the existence
of two ``disjoint'' $\vercolvertices$-saturating matchings in $B$.}

\begin{defn}[DB Graph $B^{||}$ of $\binarymatrix$] \label{def:DPGraph}
  Consider an $\dimy\times\nfactrue$-dimensional binary matrix $\binarymatrix$ and its bipartite graph $B=(\verrowvertices,\vercolvertices,E_B)$.
  The Duplicated Bipartite (DB) graph $B^{||}=(\verrowvertices^{||},\vercolvertices^{||},E_{B^{||}})$ of $\binarymatrix$ is the bipartite graph of its DB matrix $\binarymatrix^{||}$.
\end{defn}

\begin{lem}\label{lem:red4}
  Consider an $\dimy\times\nfactrue$-dimensional binary matrix $\binarymatrix$, its bipartite graph $B=(\verrowvertices,\vercolvertices,E_B)$, its DB matrix $\binarymatrix^{||}$, and its DB graph $B^{||}=(\vercolvertices^{||},\verrowvertices^{||},E_{B^{||}})$, where $\dimy\ge2\nfactrue$.
  Then, the following two tasks are equivalent in the sense that the same vertices, if any, of $\verrowvertices$ solve both:
  \begin{enumerate}[label=(\roman*)]
    \item Find two $\vercolvertices$-saturating matchings in $B$ such that the endpoints of the first matching in $\verrowvertices$ are disjoint from those of the second matching.
    \item Find a $\vercolvertices^{||}$-saturating matching in $B^{||}$.
  \end{enumerate}
\end{lem}

An instructive way to think of Lemma~\ref{lem:red4} is it being the same as Lemma~\ref{lem:red2}, but with a permutation $\rho$ applied to the rows of $\binarymatrix$.
Since $\binarymatrix^{||}$ has the same row labels as $\binarymatrix$, $\rho$ can be applied to the rows of $\binarymatrix^{||}$ as well.
This way, $\rho$ maps task (i) in Lemma~\ref{lem:red2} to task (i) in Lemma~\ref{lem:red4} and does the same for task (ii).
The inverse $\rho^{-1}$ maps the tasks back from Lemma~\ref{lem:red4} to Lemma~\ref{lem:red2}, thus closing the loop.
Here, we present a direct mechanical proof.
\begin{proof}[Proof of Lemma~\ref{lem:red4}]
  If we have two disjoint $\vercolvertices$-saturating matchings $M_1=\{\{z_{i_1},c_1\},\ldots,\{z_{i_\nfactrue},c_\nfactrue\}\}$ and $M_2=\{\{z_{j_1},c_1\},\allowbreak\ldots,\{z_{j_\nfactrue},c_\nfactrue\}\}$, then we can relabel $M_2$ to $M_2^*=\{\{z_{i_1},c^{*}_1\},\ldots,\{z_{i_\nfactrue},c^{*}_\nfactrue\}\}$.
  $M_1\cup M_2^{*}$ is a $\vercolvertices^{||}$-saturating matching in $B^{||}$ because every element of $\vercolvertices^{||}=\{c_1,\ldots,c_\nfactrue,c^{*}_1,\ldots,c^{*}_\nfactrue\}$ is covered exactly once, and the other endpoints in $\verrowvertices^{||}$ are covered at most once.

  For the inverse direction, assume that $M=\{\{z_{i_1},c_1\},\ldots,\{z_{i_\nfactrue},c_\nfactrue\},\{z_{i_{\nfactrue+1}},c^{*}_1\},\ldots,\{z_{i_{2\nfactrue}},c^{*}_\nfactrue\}\}$ is a $\vercolvertices^{||}$-saturating matching in $B^{||}$.
  Then, $M_1=\{\{z_{i_1},c_1\},\ldots,\{z_{i_\nfactrue},c_\nfactrue\}\}$ and $M_2=\{\{z_{i_{\nfactrue+1}},c_1\},\ldots,\{z_{i_{2\nfactrue}},c_\nfactrue\}\}$ are two disjoint $\vercolvertices$-saturating matchings in $B$.
\end{proof}

One final statement considers the size of a maximum matching and thus concludes one side of the aforementioned duality theorem by K\H{o}ning and Egerv\'ary, which is formally introduced later with its necessary terminology.
\begin{prop}\label{lem:red5}
  Consider an $\dimy\times\nfactrue$-dimensional binary matrix $\binarymatrix$ and its DB graph $B^{||}$, where $\dimy\ge2\nfactrue$.
  Then, the size of a matching in $B^{||}$ is at most $2\nfactrue$.
  Furthermore, if the size of a maximum matching is \SFS{equal to} $2\nfactrue$, then $\RDro$ holds for all $\facloadmatrix$ generated by $\binarymatrix$, except for a set of Lebesgue-measure zero.
\end{prop}
\begin{proof}[Proof of Proposition~\ref{lem:red5}]
  Matchings in $B^{||}$ only go between $\verrowvertices^{||}$ and $\vercolvertices^{||}$.
  Since, by definition, $\vercolvertices^{||}$
  is of size $2\nfactrue$, all matchings in $B^{||}$ contain at most $2\nfactrue$ edges.

  If there is a matching of size $2\nfactrue$, then it is a $\vercolvertices^{||}$-saturating matching.
  Lemma~\ref{lem:red4} then implies that there are two disjoint $\vercolvertices$-saturating matchings in $B$.
  Finally, applying Lemma~\ref{lem:red3} to these two matchings gives two disjoint $\nfactrue\times\nfactrue$-dimensional invertible submatrices in $\facloadmatrix$, and therefore $\RDro$ holds.
\end{proof}

Now, we place $\CRro$ onto the other side of the duality.
$\CRro$ is a statement about columns of $\binarymatrix$ being connected to sufficient number of its rows via ones.
Below, that notion of sufficiency is translated into the language of bipartite graphs.
We show that $\CRro$ is sufficient for $B^{||}$ to have many edges in a specific sense.
So many that one cannot do better than pick the entire $\vercolvertices^{||}$ as a set of vertices to cover all edges $E_{B^{||}}$.
In the following, we define vertex covers for bipartite graphs and present the duality theorem.

\begin{defn}[Vertex Cover] \label{def:vertexcover}
  A vertex cover in a bipartite graph $B=(\verrowvertices,\vercolvertices,E_B)$ is a set of vertices $C\subseteq\verrowvertices\cup\vercolvertices$ such that every edge in $E_B$ has
  at least one endpoint
  in $C$.
  A minimum vertex cover in $B$ is a vertex cover with the minimal set size among all vertex covers in $B$.\footnote{%
    There may exist vertex covers in $B$ that cannot be reduced to smaller vertex covers but are not minimum vertex covers, which complicates the search for a minimum.
    In Figure~\ref{fig:csm}, $\{c_1,c_2,c_3\}$ is not minimal but cannot be reduced further.
  }
\end{defn}

\begin{othersthm}[K\H{o}nig~\cite{kon:gra} and Egerv\'ary~\cite{ege:mat}]
  The size of a maximum matching is equal to the size of a minimum vertex cover in bipartite graphs.
\end{othersthm}

\SFS{For illustration, consider  Figure~\ref{fig:csm}.}
On the right hand side, $\{c_1,c_2,c_3\}$ is a vertex cover in $B$ because all edges in $E_B$ touch at least one of these vertices.
There is a smaller vertex cover: the set $\{z_1,c_1\}$.
This vertex cover has size two, which is also the size of a matching in $B$, shown with thick edges.
According to the duality theorem by K\H{o}nig and Egerv\'ary, said matching is therefore a maximum matching and $\{z_1,c_1\}$ a minimum vertex cover.
With the help of the duality theorem and Proposition~\ref{lem:red5}, it only remains to show that \SFS{the DB graph $B^{||}$ of $\delta$} has a minimum vertex cover of size $2\nfactrue$ if $\binarymatrix$ satisfies $\CRro$.
In \SFS{this case,} %
the duality theorem implies the existence of a matching of size $2\nfactrue$ in $B^{||}$ %
and, \SFS{consequently,} %
Proposition~\ref{lem:red5} implies that all factor loading matrices $\facloadmatrix$ generated by $\binarymatrix$ satisfy $\RDro$, except for a set of Lebesgue measure zero.

\SFS{The following %
Lemma~\ref{lem:cr1} resembles a counting rule $\CRro$ for DB matrices
and is instrumental in characterizing vertex covers in $B^{||}$.
Lemma~\ref{lem:cr1}  is employed to prove the Proposition~\ref{lem:cr2}, which is the final piece required for the proof of Theorem~\ref{thm:suf}.}

\begin{lem}\label{lem:cr1}
  Assume that an $\dimy\times\nfactrue$ binary matrix $\binarymatrix$ satisfies $\CRro$.
  Then, any subset of $1\le q\le2\nfactrue$ columns of the DB matrix $\binarymatrix^{||}$ contains at least $q$ nonzero rows.
\end{lem}

\begin{proof}[Proof of Lemma~\ref{lem:cr1}]
  Consider a submatrix $\iota$ of $q$ columns of $\binarymatrix^{||}$.
  Denote by $l\ge q/2$ the number of unique (either original or duplicate) columns in $\iota$.
  More formally, $l$ is the largest number such that $l$ columns of $\iota$ (potentially reordered) are equal to a submatrix of $l$ columns of $\binarymatrix$.
  Then, these $l$ columns contain at least $2l$ nonzero rows due to $\CRro$, and hence $\iota$ contains at least $2l\ge 2q/2=q$ nonzero rows.
\end{proof}

\begin{prop}\label{lem:cr2}
  If an $\dimy\times\nfactrue$ binary matrix $\binarymatrix$ satisfies $\CRro$, then its DB graph $B^{||}$ has a minimum vertex cover of size $2\nfactrue$.
\end{prop}

\begin{proof}[Proof of Proposition~\ref{lem:cr2}]
  To start, note that $\dimy\ge2\nfactrue$, and $C=\vercolvertices$ is a vertex cover in $B^{||}$ with $2\nfactrue$ vertices, so $2\nfactrue$ is an attainable upper bound for the minimum.

  Assume that $\CRro$ holds for $\binarymatrix$, and consider a vertex cover $C$ in its DB graph $B^{||}$.
  $C$ contains in total $k$ vertices: without loss of generality, let us assume that these are $c_1,\ldots,c_{k_1}$, $c^{*}_1,\ldots,c^{*}_{k_2}$, and $z_1,\ldots,z_{k_3}$, where $k=k_1+k_2+k_3$ and $k_1\ge k_2\ge0$ and $k_3\ge0$.

  Now, we consider all edges $\mathcal{E}\subseteq E_{B^{||}}$ in $B^{||}$ that are not covered by column vertices $C\cap\vercolvertices=\{c_1,\ldots,c_{k_1},\allowbreak c^{*}_1,\ldots,\allowbreak c^{*}_{k_2}\}$ but only by row vertices $C\cap\verrowvertices=\{z_1,\ldots,\allowbreak z_{k_3}\}$.
  Columns of $\binarymatrix^{||}$ that correspond to vertices $\vercolvertices\setminus C=\{c_{k_1+1},\ldots,c_\nfactrue,\allowbreak c^{*}_{k_2+1},\ldots,c^{*}_\nfactrue\}$ collect the nonzero entries that correspond to $\mathcal{E}$.
  Denote by $\iota$ these columns of $\binarymatrix^{||}$.
  For completeness, the exact column indices are $k_1+1,\ldots,\nfactrue,\nfactrue+k_2+1,\ldots,2\nfactrue$.
  Then, $\iota$ contains $2\nfactrue-(k_1+k_2)$ columns, and, due to Lemma~\ref{lem:cr1}, $\iota$ contains at least $2\nfactrue-(k_1+k_2)$ nonzero rows.
  But the $k_3$ row vertices $C\cap\verrowvertices$ have to cover all these rows, due to how $C$, $\iota$ and $k_3$ are defined, so $k_3\ge2\nfactrue-(k_1+k_2)$.
  This implies that $k=k_1+k_2+k_3\ge2\nfactrue$, and the proof is complete.
\end{proof}

\begin{thm}\label{thm:suf}
  Let $\binarymatrix$ be a binary matrix of size $\dimy\times\nfactrue$ and $s\ge0$ a nonnegative integer, where $\dimy\ge2\nfactrue+s$.
  Then, the following statements hold:
  \begin{enumerate}[label=(\roman*)]
    \item If $\binarymatrix$ violates the counting rule $\CRrs$, then the row deletion property $\RDrs$ is violated for all $\facloadmatrix$ generated by $\binarymatrix$.
    \item If $\binarymatrix$ satisfies the counting rule $\CRrs$, then the row deletion property $\RDrs$ holds for all $\facloadmatrix$ generated by $\binarymatrix$, except for a set of Lebesgue-measure zero.
  \end{enumerate}
\end{thm}
\begin{proof}[Proof of Theorem~\ref{thm:suf}]
  Necessity of $\CRrs$ is a direct consequence of Sato's theorem Theorem 3.3~\cite{sat:stu}.
  The theorem states that if any rotation $\facloadmatrix G$ of $\facloadmatrix$ by a non-singular $\nfactrue\times\nfactrue$-matrix $G$ violates $\CRrs$, then $\facloadmatrix$ also violates $\RDrs$.
  If $\binarymatrix$ violates $\CRrs$, and $\facloadmatrix$ is generated by $\binarymatrix$, then, by setting $G=I_{\nfactrue}$, we have that $\facloadmatrix$ violates $\RDrs$.

  Propositions~\ref{lem:red5} and~\ref{lem:cr2} together imply that $\RDro$ holds for all $\facloadmatrix$ generated by $\binarymatrix$ if $\binarymatrix$ satisfies $\CRro$, except for a set of Lebesgue measure zero, and Proposition~\ref{lem:s0} generalizes the result to $\RDrs$ and $\CRrs$.
\end{proof}

  We %
  conclude the section with a corollary that is applied in Section~\ref{sec:illustration} to identify variance identified models in sparse Bayesian factor analysis.

\begin{cor}\label{cor:varide}
  If an $\dimy\times\nfactrue$-dimensional binary matrix $\binarymatrix$ satisfies $\CRr$, where $\dimy\ge2\nfactrue+1$, then model~\eqref{eq:model} with any factor loading matrix $\facloadmatrix$ generated by $\binarymatrix$ is variance identified, except for a \SFS{set of} Lebesgue measure zero.
\end{cor}
The link from $\facloadmatrix$ to variance identification is the Anderson-Rubin theorem~\cite{and-rub:sta} stating that $\Sigmaeps$ is identified if $\facloadmatrix$ satisfies $\RDr$.
Note, however, that the Anderson-Rubin theorem and thus $\CRr$ are not necessary for variance identification even in sparse BFA.
Appendix~\ref{sec:example} provides an example of a sparse variance identified model that does not satisfy $\RDr$.

In order to apply Corollary~\ref{cor:varide} in practice, one needs to verify $\CRr$ for a given binary matrix $\binarymatrix$.
We establish the applicability of Corollary~\ref{cor:varide} for large Bayesian sparse factor models in the next section.

\subsection{Verifying Variance Identification}\label{sec:fast}

\def\mwvcmin{M^\star}
\def\mincut{\kappa^\star}

We extend the previous section and describe an efficient algorithm that verifies $\CRr$.
An initial idea might be to visit all the nonempty submatrices of $\binarymatrix$ that consist of $q$ columns for $1\le q\le\nfactrue$ and count the number of nonzero rows.
However, that approach examines $2^\nfactrue-1$ matrices, which is computationally infeasible for large $\nfactrue$: in \SFS{Bayesian inference,} %
where $\binarymatrix$ is sampled from \SFS{the posterior
distribution and %
many binary matrices need} to be checked, this step may \SFS{induce considerable} %
computational cost.
In this section, we develop a representation of the verification task in graph theory that helps us to prove our second main result: a feasible algorithm for the verification of $\CRr$ \SFS{even for large $\dimy$}.
Formally, we show in Theorem~\ref{thm:ver} that $\CRr$ can be verified in a number of steps that is polynomial in $\dimy$ and $\nfactrue$.
In this section, we provide a constructive proof of the theorem, which can be implemented in practice to verify variance identification in sparse BFA based on Corollary~\ref{cor:varide}.

\SFS{At this point, a few remarks concerning zero rows and zero columns in $\binarymatrix$ are in order. Obviously, }
zero columns are not allowed in $\CRr$ matrices.
\SFS{On the other hand, zero rows might be present and} can be removed from $\binarymatrix$ without loss of generality.
In particular, the addition or removal of zero rows does not influence rank conditions %
and $\RDr$ holds for $\facloadmatrix$ if and only if it holds for $\facloadmatrix$ without its zero rows.
Henceforth, we assume that every row and column of $\binarymatrix$ has at least one nonzero element.

Now, we introduce an extended notion of bipartite graphs that allows us to represent the verification of $\CRr$ in a graph-theoretical framework.

\begin{defn}[Weighted Bipartite Graph and Minimum Weighted Vertex Cover (MWVC)] \label{def:weightedgraph}
  A vertex-weighted (henceforth, simply weighted) bipartite graph $B=(\verrowvertices,\vercolvertices,E_B,w_B)$ is a bipartite graph with a weight mapping $w_B:\verrowvertices\cup\vercolvertices\to\mathbb{N}$.
  A minimum weighted vertex cover in $B$ is a vertex cover $C$ with the minimal total weight $\mwvcmin$ among all vertex covers in $B$, where the total weight $w_B(C)$ of the vertex cover is the sum of the weights of the vertices in the cover.
\end{defn}

Now we present Propositions~\ref{lem:blp} and~\ref{lem:alg}, which constitute the two pieces for the proof of Theorem~\ref{thm:ver}
In Proposition~\ref{lem:blp}, we design a weighted bipartite graph $B$ such that the verification of $\CRr$ on $\binarymatrix$ is equivalent to computing the total weight $\mwvcmin$ of the MWVC on $B$.
Finally, in Proposition~\ref{lem:alg}, we show that the MWVC at hand can be solved efficiently via a polynomial algorithm.

Then, the following proposition provides the basis for the polynomial algorithm.
The intuition behind the vertex cover is that the submatrix formed by the rows and columns that are left out is a zero matrix in $\binarymatrix$.

\begin{prop}\label{lem:blp}
  Consider an $\dimy\times\nfactrue$-dimensional binary matrix $\binarymatrix$.
  Let $\vergraph=(\verrowvertices,\vercolvertices,\veredges,w_B)$ be the bipartite graph of $\binarymatrix$ equipped with weights: define $w_B(z_i)=\nfactrue$ for $z_i\in\verrowvertices$ and $w_B(c_j)=2\nfactrue+1$ for $c_j\in\vercolvertices$.
  Then, $\binarymatrix$ satisfies $\CRr$ if and only if the total weight $\mwvcmin$ of the MWVC in $\vergraph$ is at least $\nfactrue(2\nfactrue+1)$.
\end{prop}

\begin{proof}[Proof of Proposition~\ref{lem:blp}]
  First, note that $\mwvcmin\le\nfactrue(2\nfactrue+1)$ always holds.
  Indeed, the set $\vercolvertices$ has total weight $(2\nfactrue+1)\nfactrue + \nfactrue\cdot0$, and it is a vertex cover.
  Now we turn to the statement.

  For the first direction of the proof, assume that $\CRr$ does not hold; i.e., there exists a submatrix $\binarymatrix_q$ made of $1\le q\le\nfactrue$ columns of $\binarymatrix$ with at most $2q$ nonzero rows.
  Then $\mwvcmin<\nfactrue(2\nfactrue+1)$.
  Indeed, vertices corresponding to the $2q$ nonzero rows of $\binarymatrix_q$ and the $r-q$ columns outside of $\binarymatrix_q$ constitute a vertex cover.
  In this case, $\mwvcmin\le(2\nfactrue+1)(\nfactrue-q)+\nfactrue\cdot2q=\nfactrue(2\nfactrue+1)-q<\nfactrue(2\nfactrue+1)$.

  For the opposite direction, assume that $\CRr$ holds. We show that $\mwvcmin$ always takes at least the aforementioned value $\nfactrue(2\nfactrue+1)$.
  Let us take any vertex cover $C$ and denote by $k$ and $l$ the number of columns and, respectively, rows that correspond to vertices included in the vertex cover.
  For $k=\nfactrue$, the total weight is at least
  \SFS{equal to} $\nfactrue(2\nfactrue+1)$.
  Now, consider $0\le k\le\nfactrue-1$.
  There are $\nfactrue-k$ columns in $\binarymatrix$ that correspond to vertices excluded from the vertex cover; the submatrix constructed from these columns contains at least $2(\nfactrue-k)+1$ nonzero rows due to $\CRr$.
  Hence, in order to cover the vertices corresponding to these nonzero rows, we must have $l\ge2(\nfactrue-k)+1$.
  This means that the total weight $w_B(C)$ for this setting evaluates to $(2\nfactrue+1)k+\nfactrue l\ge(2\nfactrue+1)k+\nfactrue(2(\nfactrue-k)+1)=\nfactrue(2\nfactrue+1)+k\ge\nfactrue(2\nfactrue+1)$.
  Since our argument holds for all vertex covers, we have shown that $\mwvcmin\ge\nfactrue(2\nfactrue+1)$.
\end{proof}

In the next statement, we use the Big-$O$ notation to describe the computational complexity of the algorithm.
\begin{prop}\label{lem:alg}
  The MWVC in $\vergraph$ and its total weight $\mwvcmin$ can be computed in $O(P(\nfactrue,\dimy))$ steps, where $P(\nfactrue,\dimy)$ is a polynomial in $\nfactrue$ and $\dimy$.
\end{prop}

See the proof below.
We do not directly work on the weighted bipartite graph to find the MWVC, but we rather reformulate the problem as a minimal network cut problem and refer to known solutions for that problem, such as Dinic's algorithm~\cite[chapter 8]{tar:dat}.
Therefore, in order to present the reformulation, we first introduce some notions from network theory.
Even though a network is also a graph, we deliberately use different terminology for its parts to improve readability.
In particular, we use ``node'' instead of ``vertex'' and ``arrow'' instead of ``edge''.
We denote arrows as tuples by round brackets, e.g., $(c_2,z_1)$, because they are directed and the order of the nodes matters in networks, in contrast to the set-notation of the curly brackets $\{z_1,c_2\}$ used for undirected edges in all graphs in this paper.
\begin{defn}[Network and Cut in a Network]
  A network $N=(V,E,\kappa)$ is a set of nodes $V$ combined with a set of arrows $E\in V\times V$, which are ordered pairs of nodes.
  Furthermore, networks always have two distinguished nodes: the source node $s\in V$ and the sink node $t\in V$, using common notation.\footnote{%
    One can imagine a network as a model for a pipe system.
    The source node is the water source, the sink node is the water drain, and the arrows are pipes connecting the nodes.
    Each pipe has a diameter, which fixes the capacity of the pipe.
    A simplistic cut naturally arises if many pipes get clogged such that the sink node is cut away from the source node; technically, the cut is then the set of nodes that still get water from the source.
    Note, however, that the definition of a cut in a network is more general than this intuitive description.
  }
  Each arrow $(u,v)\in E$ going from node $u\in V$ to node $v\in V$ has capacity $\kappa(u,v)\in\mathbb{N}\cup\infty$.
  A cut $C\subset V$ is a set of nodes such that $s\in C$ and $t\notin C$.
  The capacity $\kappa(C)$ of a cut is the sum of the capacities of the arrows that start in $C$ and end outside of $C$ in $V\setminus C$.
  A minimal cut is a cut whose capacity $\mincut$ is minimal among all cuts in the network.
\end{defn}

\SFS{For illustration, consider Figure~\ref{exnetwork}.}
 The  \SFS{upper part} %
 shows an example of the weighted bipartite graph $B$ for a $8\times3$-dimensional $\binarymatrix$.
The bottom half of Figure~\ref{exnetwork} shows an example of a network $N=(V,E,\kappa)$
created from $B$ as described below in the proof of Proposition~\ref{lem:alg}, where $V=\{s,c_1,c_2,c_3,z_1,z_2,z_3,z_4,\allowbreak z_5,z_6,z_7,z_8,t\}$. There are twenty-four edges in $E$, and capacities are $\kappa(s,c_i)=7$, $\kappa(c_i,z_j)=\infty$ if $\binarymatrix_{j,i}=1$ and otherwise $0$, and $\kappa(z_j,t)=3$.
A cut can be, for example, $C=\{s,c_1,z_1,z_3,z_4,z_5\}$ with capacity $\kappa(C)=2\cdot7+0+4\cdot3=26$.

\begin{proof}[Proof of Proposition~\ref{lem:alg}]
We first construct the network $N$ from $B$.\footnote{This network construction is inspired by lecture notes in \cite{goe:adv}.}
Its nodes $V$ are the source node $s\in V$, one node $z_i^N$ for every vertex $z_i$ of $B$, one node $c_j^N\in V$ for every vertex $c_j$ of $B$, and the sink node $t$.
Henceforth, to reduce clutter, we use the same notation for the nodes of $N$ as for the vertices of $B$ and the rows and columns of $\binarymatrix$: ``row'' $z_i$ is part of $\binarymatrix$, ``vertex'' or ``row vertex'' $z_i$ is part of $B$, and ``node'' or ``row node'' $z_i$ is part of $N$; the same applies to columns, (column) vertices and (column) nodes denoted by $c_j$.
Continuing with the construction of $N$, there are three groups of arrows: for every column node $c_j$, an arrow goes from $s$ to $c_j$ with capacity $2\nfactrue+1$; for every edge $\{z_i,c_j\}$ in $B$, an arrow $(c_j,z_i)$ goes from node $c_j$ to node $z_i$ with infinite capacity; and for every row node $z_i$, an arrow goes from $z_i$ to $t$ with capacity $\nfactrue$.
Figure~\ref{exnetwork} shows an example of the construction with $\dimy=8$ and $\nfactrue=3$ including $\binarymatrix$, $B$, and $N$.

\tikzset{%
  every neuron/.style={
    circle,
    draw,
    minimum size=1cm
  }
}

\def\stepsize{0.3cm}
\begin{figure}[!t]
  \centering
  \begin{tikzpicture}[x=0cm, y=0cm, >=stealth]
    \node [matrix] (delta) at (-10, 0) {
        \node (delta corner) {}; & \node [inner xsep=0pt] (c1) {$c_1$}; & \node {$c_2$}; & \node [inner xsep=0pt] (c3) {$c_3$}; \\
      \node [inner ysep=0pt] (r1) {$z_1$}; & \node {1}; & \node {0}; & \node {0}; \\
      \node {$z_2$}; & \node {0}; & \node {1}; & \node {0}; \\
      \node {$z_3$}; & \node {1}; & \node {1}; & \node {0}; \\
      \node {$z_4$}; & \node {1}; & \node {0}; & \node {1}; \\
      \node {$z_5$}; & \node {1}; & \node {1}; & \node {1}; \\
      \node {$z_6$}; & \node {0}; & \node {0}; & \node {1}; \\
      \node {$z_7$}; & \node {0}; & \node {1}; & \node {1}; \\
      \node [inner ysep=0pt] (r8) {$z_8$}; & \node (c1r8) {0}; & \node {1}; & \node (c3r8) {0}; \\
    };
    \node [venn diagram, fit=(c1) (c3)] {};
    \node [venn diagram, rotate fit=90, fit=(r1) (r8)] {};

    \node (hiff) [right of=delta, xshift=1.3cm] {$\qquad\longmapsto\qquad$};
    \begin{scope}[right of=hiff, xshift=3.5cm]
      \begin{scope}[xshift=1.5*\stepsize, yshift=4*\stepsize, start chain=going right, node distance=\stepsize]
        \node foreach \i in {1,2,3} [on chain] (fg\i) {$c_{\i}$};
      \end{scope}
      \begin{scope}[xshift=1*\stepsize, yshift=-4*\stepsize, start chain=going right, node distance=\stepsize]
        \node foreach \i in {1,2,...,4} [on chain] (sg\i) {$z_{\i}$};
      \end{scope}
      \begin{scope}[yshift=-7*\stepsize, start chain=going right, node distance=\stepsize]
        \node foreach \i in {5,6,...,8} [on chain] (sg\i) {$z_{\i}$};
      \end{scope}
      \node [right of=sg4, xshift=.0cm, yshift=2.35cm] {$w_B=7$};
      \node [right of=sg4, xshift=.8cm, yshift=-0.5cm] {$w_B=3$};
    \end{scope}

    \node (fg) [venn diagram, fit=(fg1) (fg3)] {};
    \node (sg) [venn diagram, fit=(sg1) (sg4) (sg5) (sg8)] {};

    \begin{scope}
      \draw [edge] (fg1) -- (sg1);
      \draw [edge] (fg1) -- (sg3);
      \draw [edge] (fg1) -- (sg4);
      \draw [edge] (fg1) -- (sg5);
      \draw [edge] (fg2) -- (sg2);
      \draw [edge] (fg2) -- (sg3);
      \draw [edge] (fg2) -- (sg5);
      \draw [edge] (fg2) -- (sg7);
      \draw [edge] (fg2) -- (sg8);
      \draw [edge] (fg3) -- (sg4);
      \draw [edge] (fg3) -- (sg5);
      \draw [edge] (fg3) -- (sg6);
      \draw [edge] (fg3) -- (sg7);
    \end{scope}
  \end{tikzpicture}

  \begin{tikzpicture}[x=1.5cm, y=1.5cm, >=stealth]
    \node (iff) at (4.0, 4.2) {$\Big\Updownarrow$};

    \node [every neuron] (source) at (0,0) {$s$};

    \foreach \m [count=\y] in {1,2,3}
    \node [every neuron] (factor-\m) at (2,2-\y) {$c_\y$};

    \foreach \m [count=\y] in {1,2,3,4,5,6,7,8}
    \node [every neuron] (observation-\m) at (4,4.5-\y) {$z_\y$};

    \node [every neuron] (sink) at (6,0) {$t$};

    \node at (1,-1) {$7$};

    \node at (3,-2.5) {$\infty$};

    \node at (5,-2.5) {$3$};

    \foreach \l in {1,2,3}
    \draw [->] (source) -- (factor-\l);

    \foreach \l in {1,2,3,4,5,6,7,8}
    \draw [->] (observation-\l) -- (sink);

    \foreach \l in {1,3,4,5}
    \draw [->] (factor-1) -- (observation-\l);

    \foreach \l in {2,3,5,7,8}
    \draw [->] (factor-2) -- (observation-\l);

    \foreach \l in {4,5,6,7}
    \draw [->] (factor-3) -- (observation-\l);

  \end{tikzpicture}
  \caption{A \SFS{binary matrix} %
  $\binarymatrix$  with $\nfactrue=3$ and $\dimy=8$ (top left), its bipartite graph $\vergraph$ extended with weights \SFS{$w_B$} (top right),  and a network $N$ created from $B$ as described in \SFS{the proof of} Proposition~\ref{lem:alg} (bottom).
  Labels $z_i$ and $c_j$ correspond to the columns and rows of $\binarymatrix$, respectively, $i=1,\ldots,8$, $j=1,2,3$.
  In $\vergraph$, vertices of $\vercolvertices$ and $\verrowvertices$ have weight $7$ and $3$, respectively.
  In $N$, edges between $s$ and $\vercolvertices$, between $\vercolvertices$ and $\verrowvertices$, and between $\verrowvertices$ and $t$ have capacity $7$, $\infty$, and $3$, respectively.}
  \label{exnetwork}
\end{figure}

We prove that the capacity $\mincut$ of the minimal cut in $N$ equals the total weight $\mwvcmin$ of the MWVC in $B$.
Then, we are finished, since $\mincut$ can be computed in $O(P(\nfactrue,\dimy))=O((\dimy+\nfactrue+2)^2(\dimy+\nfactrue+\dimy\nfactrue))$ steps using Dinic's algorithm~\cite{tar:dat}, where $O$ is the Big-$O$ notation.
To see this, observe that $N$ has $\lvert V\rvert=\dimy+\nfactrue+2$ nodes and at most \SFS{$ \lvert E\rvert \le \dimy+\nfactrue+\dimy\nfactrue$} arrows, with the maximum attained only if $\binarymatrix$ is a full binary matrix, and Dinic's algorithm solves the task with a computational complexity of $O(\lvert V\rvert\lvert E\rvert^2)$.

The workhorse of our proof is a bijection between sets of vertices in $B$ and sets of nodes in $N$.
This bijection has special behavior for our construction of $N$, namely, it is also a bijection between vertex covers in $B$ and finite-capacity cuts in $N$. %
\SFS{This allows us to}
efficiently find the MWVC in $B$ by computing the minimal cut in $N$ using Dinic's algorithm.
Denote by $S\subseteq\verrowvertices\cup\vercolvertices$ any set of vertices, which we map to a cut $C$ in network $N$: for every \emph{included} row vertex $z_i\in S\cap\verrowvertices$ (respectively, \emph{excluded} column vertex $c_j\in\vercolvertices\setminus S$), the row node $z_i$ (respectively, column node $c_j$) is included in $C$, and also $s$ is included in $C$, and nothing else.
$C$ is a cut since it includes $s$ and excludes $t$.
Three statements remain: first, $S$ is a vertex cover iff $C$ has finite capacity; second, the weight of $S$ equals the capacity of $C$ if it is finite; and, third, a technicality
\SFS{stating that} $\mwvcmin$ is well-defined in $B$.
Then, we know that the minimum capacity $\mincut$ is equal to the minimum weight $\mwvcmin$.

The first statement has two directions.
For the first direction, assume that $C$ has infinite capacity, which happens if and only if there is an arrow $(c_j,z_i)$ that leaves $C$, i.e., $c_j\in C$ and $z_i\notin C$.
Due to our construction of $C$, \SFS{this} %
implies that there \SFS{exists} %
an edge $\{z_i,c_j\}$ in $B$, but
neither of its endpoints $z_i$ or $c_j$ are in $S$.
Consequently, $S$ is not a vertex cover.
For the other direction, assume that $S$ is not a vertex cover.
Then, there exists an edge $\{z_i,c_j\}$ in $B$ such that $z_i\notin S$ and $c_j\notin S$, in which case $C$ has infinite capacity.
We have shown that $S$ is a vertex cover if and only if $C$ has finite capacity.

For the second statement, we examine the capacity of $C$ when it is finite.
If $z_i\in S$, then the arrow $(z_i,t)$ leaves $C$ and contributes $\nfactrue$ to its capacity, which equals the contribution of $z_i$ to the weight of $S$.
If $c_j\in S$, then the arrow $(s,c_j)$ leaves $C$ and contributes $2\nfactrue+1$ to its capacity, which equals the contribution of $c_j$ to the weight of $S$.
There are no other edges that leave $C$.
We have shown that the capacity of $C$ is equal to the weight of $S$.

For the third and final statement, note that the MWVC is an optimum over a \SFS{specific} domain, and the domain needs to be non-empty for the MWVC to be well-defined.
The domain in question is the set of vertex covers in $B$.
There always exists at least one vertex cover, e.g., all vertices $\verrowvertices\cup\vercolvertices$ are a vertex cover, and therefore the MWVC with total weight $\mwvcmin$ exists as well.
This concludes the proof.
\end{proof}

\begin{thm}\label{thm:ver}
  Property $\CRr$ can be verified algorithmically in $O(P(\nfactrue,\dimy))$ steps, where $O$ is the Big-$O$ notation, and $P(\nfactrue,\dimy)$ is a polynomial in $\nfactrue$ and $\dimy$.
\end{thm}
\begin{proof}[Proof of Theorem~\ref{thm:ver}]
  The proof follows from Propositions~\ref{lem:blp} and~\ref{lem:alg}.
\end{proof}

In the proof of Proposition~\ref{lem:alg}, we also find that the number of steps increases with $P(\nfactrue,\dimy)=(\dimy+\nfactrue+2)^2(\dimy+\nfactrue+\dimy\nfactrue)$.
For fixed $\nfactrue$, the computational complexity of our method is therefore $O(\dimy^3)$, and, for fixed $\dimy$, which is often the case, it is $O(\nfactrue^3)$ instead of complexity $O(2^\nfactrue)$ for the brute force search through all submatrices.

Although Theorem~\ref{thm:ver} only concerns $\CRr$, the result may also be used to build an algorithm that verifies $\CRrs$ in smaller settings.
Namely, it is easy to see that $\binarymatrix$ satisfies $\CRrs$ if and only if after removing any $s-1$ rows the remaining binary matrix satisfies $\CRr$, which we can verify efficiently.
That realization gives rise to a recursive algorithm with complexity $O(\dimy^{s-1} P(\nfactrue,\dimy))$, which may be practical for $s=2$ or $s=3$ for small $\dimy$.

Finally, note that, in its current form, the proof cannot be extended to a polynomial complexity $P(\nfactrue,\dimy,s)$ algorithm also in $s$ for $\CRrs$ by choosing different weights for $\verrowvertices$ or $\vercolvertices$.
In particular, if $x$ denotes the ratio of vertex weights in $\vercolvertices$ and $\verrowvertices$ (i.e., $x=(2\nfactrue+1)/\nfactrue$ above), then it can be shown that $x>2+s-1$ and $x\le2+s/\nfactrue$ are both necessary for the proof of Proposition~\ref{lem:blp} and thus for Theorem~\ref{thm:ver}.
This interval is non-empty only if $s\le1$.

\SFS{In concluding we note that implementations of this algorithm are available in R and MATLAB at \url{https://hdarjus.github.io/sparvaride/}}.

\section{Numerical Illustration} \label{sec:illustration}

\def\rmax{{H}}

We demonstrate that missing variance identification may unnecessarily inflate the estimated number of factors during exploratory factor analysis (EFA).
We choose the Bayesian paradigm, which allows us to emulate matrix sparsity using a spike-and-slab prior distribution on $\facloadmatrix$ (\SFS{to be} introduced in Section~\ref{sec:model}) and to consider variance identification as a domain restriction on that prior distribution.
Consequently, we can estimate the posterior distribution via a Markov chain Monte Carlo (MCMC) sampler under the unrestricted prior and apply the domain restriction as a post-processing step by discarding the unsatisfactory draws.
The model, its estimation, a simulation study, and a real data study are detailed below.

\subsection{Model and Prior} \label{sec:model}

To facilitate variance identification through $\CRr$, we follow the tradition of~\cite{wes:bay_fac} and introduce indicator variables $\binarymatrix_{ij}\in\{0,1\}$ for every factor loading $\facload{ij}$ as parameters to estimate for $i=1,\ldots,\dimy$, and $j=1,\ldots,\rmax$, and collected in the $\dimy\times\rmax$ matrix $\binarymatrix=(\binarymatrix_{ij})$.
Following established procedures (\cite{con-etal:bay_exp,kau-sch:bay,Fruehwirth2024Sparse}), Bayesian posterior sampling is applied with a conjugate prior on $\facloadmatrix$ and $\Sigmaeps$, combined with column-wise shrinkage on the indicator \SFS{variables $\binarymatrix_{ij}$}.
For completeness, we
\SFS{provide the full hierarchical model specification by combining model~\eqref{eq:model} with a corresponding prior:}
\begin{equation}\label{eq:bayesmodel}  %
	\begin{split}
    y_t &= \facloadmatrix f_t + \epsilon_t, \\
    \epsilon_t &\sim N_\dimy(0,\Sigmaeps), \\
    f_t &\sim N_\nfactrue(0,I), \\
    \\ \\
  \end{split}
  \qquad \qquad \qquad
  \begin{split}
		\facload{ij} \mid \binarymatrix_{ij} = 0 &\equiv 0, \\
		\facload{ij} \mid \binarymatrix_{ij} = 1 &\sim N(0, \sigma_i^2), \\
		\sigma_i^2 &\sim IG(c_0, C_0), \\
		\binarymatrix_{ij} &\sim Ber(\tau_j), \\
		\tau_j &\sim B(a_0, b_0),
	\end{split}
\end{equation}
where $IG(c_0, C_0)$ denotes the inverted gamma distribution with kernel density $x^{-c_0-1}\exp(-C_0/x)\mathbbm{1}(x>0)$, $Ber(\tau_j)$ is the Bernoulli distribution with success probability $\tau_j\in(0,1)$, $B(a_0, b_0)$ is the beta distribution with kernel density $x^{a_0-1}(1-x)^{b_0-1}\mathbbm{1}(0<x<1)$, and $t=1,\ldots,T$.
The choice of $\sigma_i^2$ as the variance lets $\facload{ij}$ capture potential scaling differences between the observation series.
Moreover, two settings are considered below for the prior on $\tau_j$: following~\cite{roc-geo:fas} and~\cite{Fruehwirth2023Cusp}, the finite one-parameter beta (1PB) prior $(a_0,b_0)=(\alpha/\rmax,1)$ is chosen first, which we call shrinkage below, and the uniform prior $(a_0,b_0)=(1,1)$ is picked as an alternative for sensitivity analysis.

A potentially influential question is the choice of $\rmax$.
One solution is the use of infinite factor models, initiated by~\cite{gha-etal:bay} and popularized by~\cite{bha-dun:spa} and~\cite{leg-etal:bay}, where one theoretically lets $\rmax$ diverge to $\infty$ while cumulatively shrinking the columns a priori
\SFS{towards zero} as the column index increases.
Here, we assume\footnote{%
  Necessarily, $2\rmax+1\le\dimy=\dimy$, where $\dimy$ is the number of observation series.
  That is essential for variance identification via $\RDr$, and therefore also via $\CRr$.}
$\rmax=\min(30,\lfloor(\dimy-1)/2\rfloor)$ to both allow for parameter identification via Corollary~\ref{cor:varide} and keep Monte Carlo simulations manageable.
Notably, recently,~\cite{Fruehwirth2023Cusp} showed that our column-wise exchangeable prior $p(\binarymatrix\mid\rmax)$ in Equation~\eqref{eq:bayesmodel} is strongly related to both the framework of~\cite{gha-etal:bay}
and%
~\cite{leg-etal:bay}.

\subsection{Estimation}

Model~\eqref{eq:bayesmodel} specifies a sparse Bayesian factor model with a spike-and-slab prior on $\facloadmatrix$.
The prior $p(\facloadmatrix\mid\rmax)$ is exchangeable both row-wise and column-wise, and the elements of $\{\sigma_i^2\}$ are independent a priori, which results in an order-independent model for the observation series.
Furthermore, the choice of standard conjugate priors for $(\facloadmatrix,\{\sigma^2_i\})$ and $\{\tau_j\}$ enables simple Gibbs sampling.
See Appendix~\ref{app:mcmc} for the steps of the MCMC algorithm.

Throughout the demonstration, we compare three domain restrictions, which we implement via post-processing of the MCMC output.
Under the unrestricted scenario, variance identification as a step is ignored, and the entire output of the MCMC procedure is retained.
In the second scenario, the necessary condition for variance identification of~\cite{and-rub:sta} is applied as a post-processing step, similar to~\cite{kau-sch:bay}.
Namely, if in all columns of $\facloadmatrix$, at least three nonzero elements are present, then the MCMC draw is retained, and, otherwise, it is excluded from summaries of the posterior distribution.
In the third scenario, the sufficient condition $\CRr$ from Corollary~\ref{cor:varide} is enforced during post-processing by only keeping the MCMC draws that satisfy the condition.
In both cases, the MCMC output is filtered before proceeding further: before any
\SFS{subsequent} analysis, we discard the joint draws of \SFS{$(\facloadmatrix,\Sigmaeps,f_1, \ldots, f_T)$} when $\facloadmatrix$ does not satisfy the necessary or, respectively, the sufficient condition.

Further steps during post-processing are estimating \SFS{the number of factors} $\nfactrue$ and the covariance matrix $\Omega=\facloadmatrix\facloadmatrix^\top+\Sigmaeps$ from the filtered or unfiltered MCMC output, depending on the scenario above.
Following~\cite{Fruehwirth2023Cusp} and~\cite{Fruehwirth2023When}, we assume a potentially too large number of factors $\rmax$ and estimate the posterior distribution for
\SFS{$\nfactrue$} %
by counting the number of active columns in $\facloadmatrix$ for every MCMC draw.
Active columns of $\facloadmatrix$ are those that contain at least two nonzero elements, and zero columns are deemed inactive.
Columns with a single nonzero element are automatically transformed to zero columns during post-processing by moving the square of the single factor loading and adding \SFS{it} to the corresponding diagonal element of $\Sigmaeps$.
The reason is that these columns are actually
\SFS{spurious factors and} they capture the variance of a single observation series, as explained in~\cite{Fruehwirth2023When}.
Finally, one acquires a posterior sample for the covariance matrix by calculating $\Omega = \facloadmatrix\facloadmatrix^\top+\Sigmaeps$ for every joint draw of $(\facloadmatrix,\Sigmaeps)$.

\begin{figure}[t]
  \centering
	\includegraphics[width=\textwidth,page={1}]{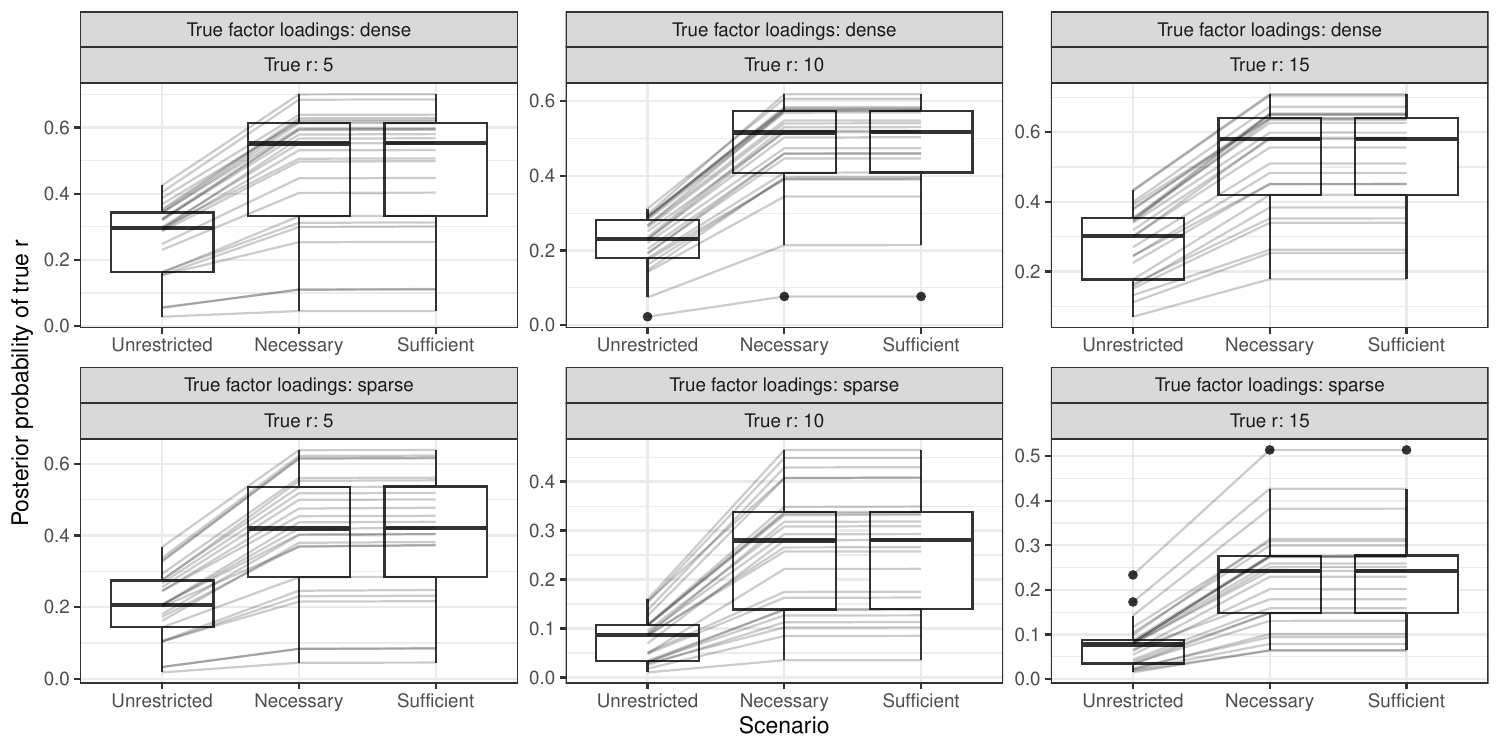} \\
  \caption{Simulation study.
  Posterior probabilities of the true number of factors are shown under the shrinkage prior \SFS{on $\tau_j$}.
  The top row corresponds to the dense setting, and the bottom row corresponds to the sparse setting.
  The columns correspond to the true number of factors in the DGP.
  All 25 repetitions are summarized in the boxplots, and the lines connect the same data set under different scenarios.}
	\label{fig:simres1}
\end{figure}

\begin{figure}[t]
  \centering
	\includegraphics[width=\textwidth,page={2}]{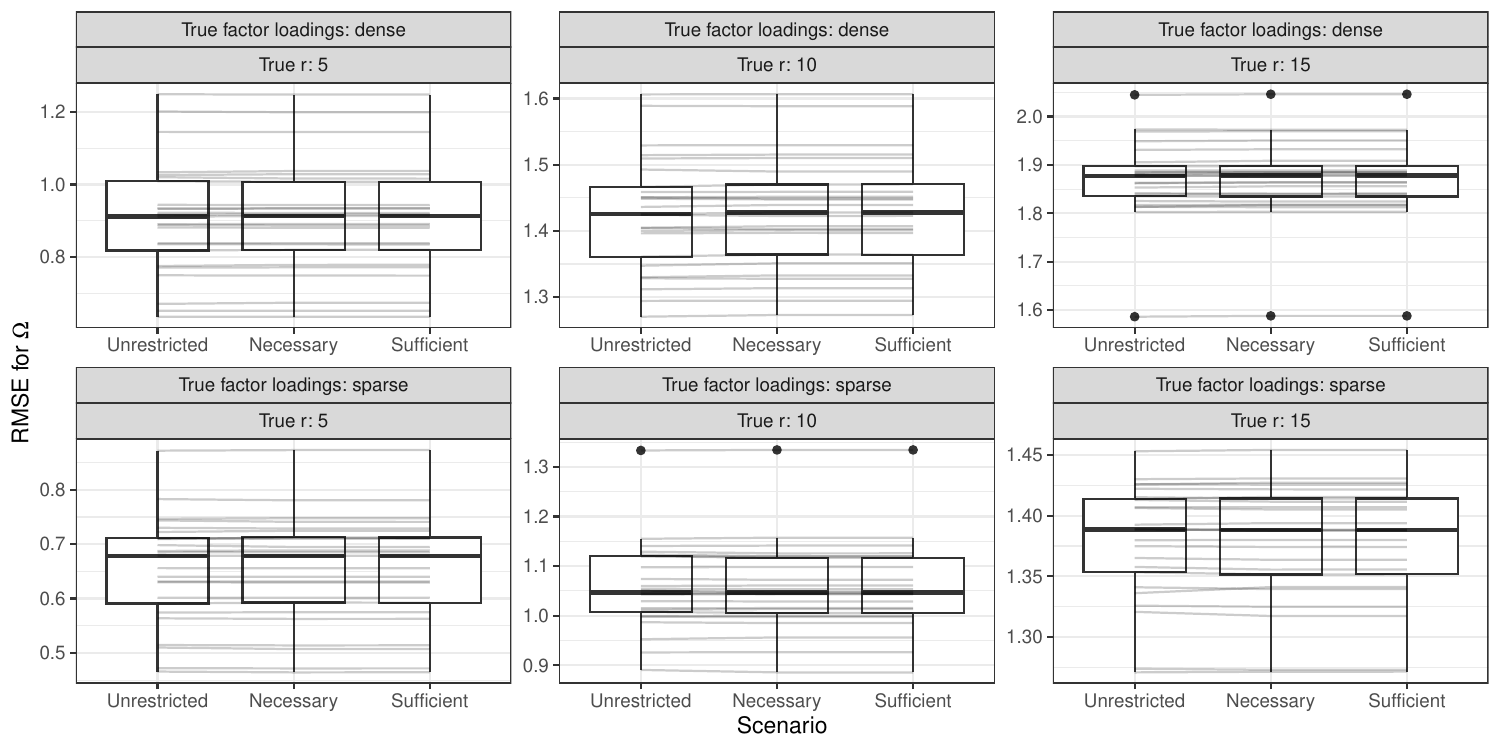}
  \caption{Simulation study.
  Root mean squared error of the estimated covariance matrix is shown under the shrinkage prior \SFS{on $\tau_j$}.
  The top row corresponds to the dense setting, and the bottom row corresponds to the sparse setting.
  The columns correspond to the true number of factors in the DGP.
  All 25 repetitions are summarized in the boxplots, and the lines connect the same data set under different scenarios.}
	\label{fig:simres2}
\end{figure}

\subsection{Simulation Study}

We follow~\cite{leg-etal:bay} and %
conduct a simulation study with three different combinations of $(\dimy,\nfactrue)$, namely, $(20,5)$, $(50,10)$, and $(100,15)$.
For each combination, $25$ repetitions of
\SFS{$T=100$} %
observations are generated.
Following~\cite{Fruehwirth2023Cusp}, we examine two settings for generating $\binarymatrix$: in the dense setting, $\binarymatrix$ is a fully nonzero \SFS{binary} matrix, and in the sparse setting, random 30\% of \SFS{the indicators in
$\binarymatrix$ are} set to zero and the remaining 70\% to one.
We always enforce the true $\binarymatrix$ to satisfy $\CRr$ by re-sampling until \SFS{this condition is met}.
In all scenarios, $\Sigmaeps$ is the identity matrix, and $\facload{ij}$ is standard normal \SFS{whenever} $\binarymatrix_{ij}$ is nonzero.

Turning to the priors, the fairly vague setting $(c_0, C_0)=(1, 0.3)$ is adopted from~\cite{leg-etal:bay}.
Finally, contrary to~\cite{Fruehwirth2023Cusp}, we do not estimate $\alpha$ to keep the model simple but rather fix $\alpha=5$, which is consistent with their findings.
Including the choice of shrinkage and uniform priors for $\tau_j$,
 300 posterior distributions are estimated in this simulation study in total.

To facilitate MCMC convergence diagnostics, four independent posterior Markov chains are simulated with distant i\-ni\-tia\-li\-za\-tions: zero, one, $\rmax-1$, and $\rmax$ randomly filled columns in $\facloadmatrix$ with standard normal draws.
In the small settings $(20,5)$ and $(50,10)$, the MCMC chains are run for $50\,000$ iterations, and the first $10\,000$ are discarded as burn-in.
However, we face significant computational challenges with our simple Gibbs sampler in the biggest setting $(100,15)$, where we run the MCMC chains for one million iterations on a cluster of 400 cores and one terrabyte memory for a total of 20 hours to see full convergence.

Figures~\ref{fig:simres1} and~\ref{fig:simres2} provide details on the results under the shrinkage prior on $\tau_j$ and follow a similar structure.
The six facets of Figure~\ref{fig:simres1} depict the posterior probability \SFS{$p(\nfactrue=\nfactrue_\text{true} \mid\yv)$} of the true number of factors, where $\yv$ are the observed data and \SFS{$\nfactrue_\text{true}$} %
is the true number of factors in the data generating process (DGP).
The first and the second rows correspond to the dense and, resp., the sparse setting, while the columns correspond to the true number of factors $\nfactrue_\text{true}$.
Within a facet, from left to right, the three boxplots summarize posterior probabilities under the unrestricted, the necessary, and, respectively, the sufficient scenario, each showing a distribution over 25 DGP repetitions.
The final ingredients of the chart are the lines that connect the corresponding repetitions, i.e., posterior summaries under different scenarios but the same data set.
The six facets of Figure~\ref{fig:simres2} depict the root mean squared error (RMSE) of the estimated covariance matrix and follow the same structure as the \SFS{six facets in Figure~\ref{fig:simres1}.}
We find that variance identification consistently reduces the estimated number of factors $\nfactrue$ without affecting the quality of the estimated covariance matrix.
In more than 50\% of the dense cases, the posterior probability of the true number of factors is below 0.5 under the unrestricted scenario but over 0.5 under both restricted scenarios, which can be seen as an important jump.
In the sparse setting, the posterior probabilities are generally lower, but the same pattern is observed.
We also find that the necessary and the sufficient scenarios yield very similar results, which we read as the necessary and sufficient conditions being almost equivalent for our DGP's.
In summary, we see that variance identification improves the estimate for the number of factors for all simulated data sets.

Results not reported here indicate the same conclusion under the uniform prior for \SFS{$\tau_j$}.
In particular, variance identification improves the estimate for the number of factors without affecting the quality of the estimated covariance matrix.
One difference is, however, that the posterior probabilities of the true number of factors are generally lower under the uniform prior than under the shrinkage prior.
While the probabilities range even up to 0.7 under the shrinkage prior, as Figure~\ref{fig:simres1} shows, the largest ones are already below 0.04 in the $(50,10)$ dense setting and below 0.001 in most of the $(100,15)$ sparse settings under the uniform prior.
The uniform prior \SFS{on $\tau_j$} does not provide as strong a signal for the correct number of factors as the shrinkage prior does, which is consistent with~\cite{Fruehwirth2023Cusp}.

\subsection{Prediction Exercise on Weekly Exchange Rate Data} \label{appEx}

Weekly returns of 17 currencies against the EUR are investigated between January, 2003, and December, 2005.
The series include the currencies of big trading partners of the Eurozone (Australian Dollar, Canadian Dollar, British Pound, Hong Kong Dollar, Japanese Yen, South Korean Won, New Zealand Dollar, Russian Ruble, Turkish Lira, and US Dollar), and important local partners (Swiss Franc, Czech Koruna, Danish Krone, Norwegian Krone, Polish Zloty, Romanian Leu, and Swedish Krona).
The chosen time period mostly avoids large international crises and heavy-tailed return distributions, as depicted in Figure~\ref{fig:data}, which renders the static latent factor model%
~\eqref{eq:model} appropriate for its analysis.

\begin{figure}[t]
	\centering
	\includegraphics[width=\linewidth]{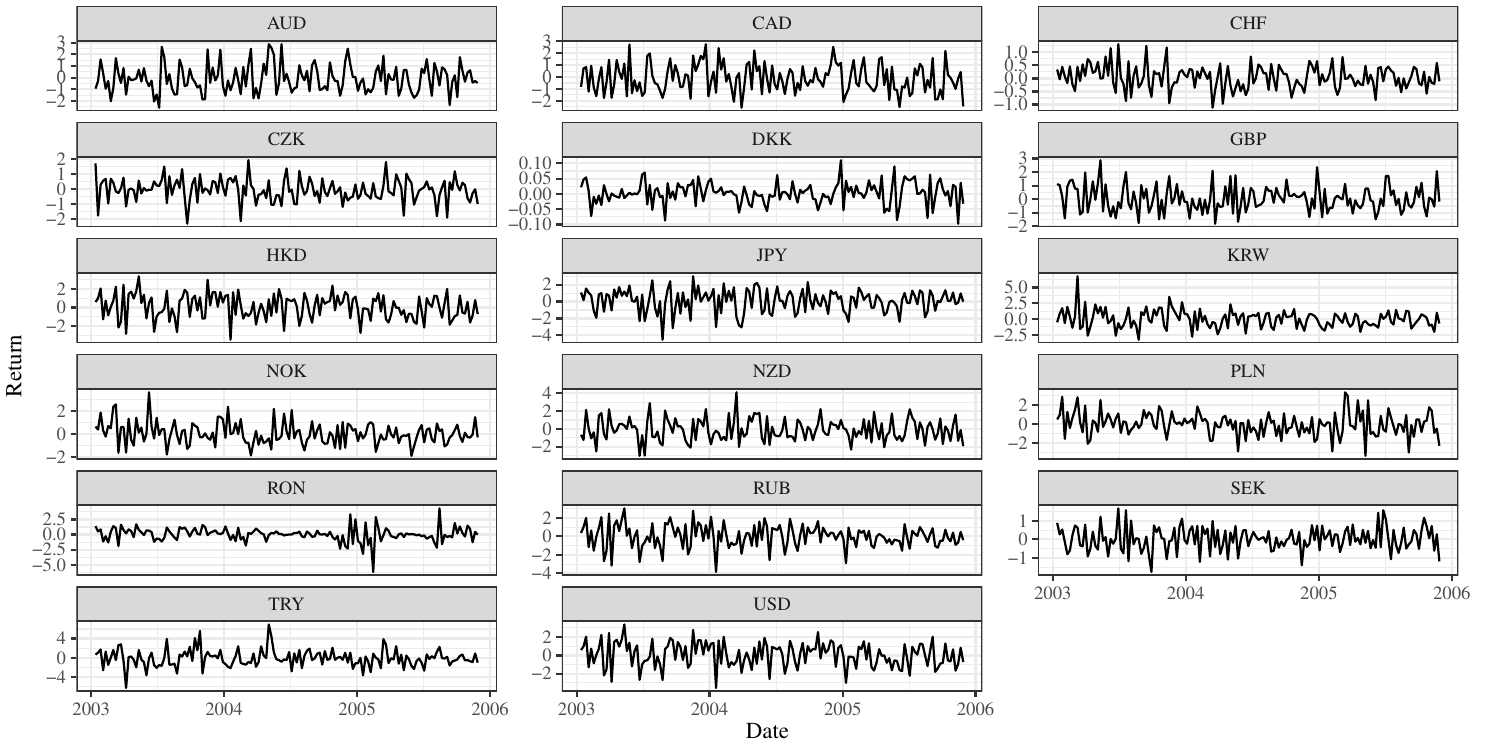}
	\caption{Real data exercise.
  Data set of 17 exchange rates against EUR.}
	\label{fig:data}
\end{figure}

Estimation is done \SFS{using a moving window of} %
52 weekly returns and the predictive performance is examined.
In particular, the log posterior predictive likelihood $\text{LPPL}=\log\int_\theta p(y_{53}\mid\theta,y_1,\ldots,y_{52}) p(\theta\mid y_1,\ldots,y_{52})\,d\theta$ of the next weekly return is estimated as the natural logarithm of the mean of the sampled posterior predictive likelihoods $\mathcal{S}=\{p(y_{53}\mid\theta,y_1,\ldots,y_{52})\}_{\theta\in\text{posterior}}$, i.e., $\log((\sum_{s\in\mathcal{S}}s)/|\mathcal{S}|)$, where $\theta$ collects all parameters of the model, and $y_t$, $t=1,\ldots,53$, are the weekly returns for a given time window, including the next weekly return $y_{53}$.
The sample means of the 52 weekly returns are subtracted from the input data before estimation and from the vector of next weekly returns before computing the LPPL.
Then, the time window is shifted by one week, and estimation and prediction are repeated.
The procedure is \SFS{repeated} %
100 times, which covers approximately two years of weekly predictions under a moving window regime.
$\rmax=8$ is chosen for this exercise, which is the largest %
$\rmax$ that satisfies $\rmax\le(\dimy-1)/2$ for $\dimy=17$, and the same priors as in the simulation study are used.
Importantly, we again consider two priors for
\SFS{$\tau_j$} (shrinkage and uniform), which results in 200 posterior distributions in total for this exercise.

During post-processing, the three scenarios
\SFS{regarding variance identification used  in the simulation study}
(unrestricted, necessary, and sufficient)
are applied to the MCMC output.
Two measures are computed for comparing the scenarios: the LPPL and the estimated number of factors.
If model $\mathcal{M}_1$ has the same LPPL as model $\mathcal{M}_2$ but fewer factors, then $\mathcal{M}_1$ is preferred for its simplicity.

The dots in Figure~\ref{fig:logliks} show the LPPL of the sufficient scenario for a moving window of width 52 relative to the LPPL of the baseline unrestricted scenario, both under the shrinkage prior on $\tau_j$.
The grey area represents the 5th to 95th percentiles of the posterior sample $\mathcal{S}$ used to estimate the LPPL under the unrestricted scenario, also relative to said baseline.
We see that the LPPL of the sufficient scenario is very close to the LPPL of the unrestricted scenario as the difference stays close to zero.
Moreover, the difference is considerably smaller than the width of the middle 90\% region of the sampling distribution of the LPPL under the unrestricted scenario.
Results not reported here show that both the uniform prior \SFS{on $\tau_j$} and the necessary scenario provide the same conclusion.
In summary, restricting the prior to \SFS{variance} identified patterns does not significantly affect predictive performance of the factor model.
Since this predictive measure purely depends on the estimated covariance matrix $\Omega$, this finding is consistent with the simulation study.

\begin{figure}[t]
	\centering
\includegraphics[width=\linewidth]{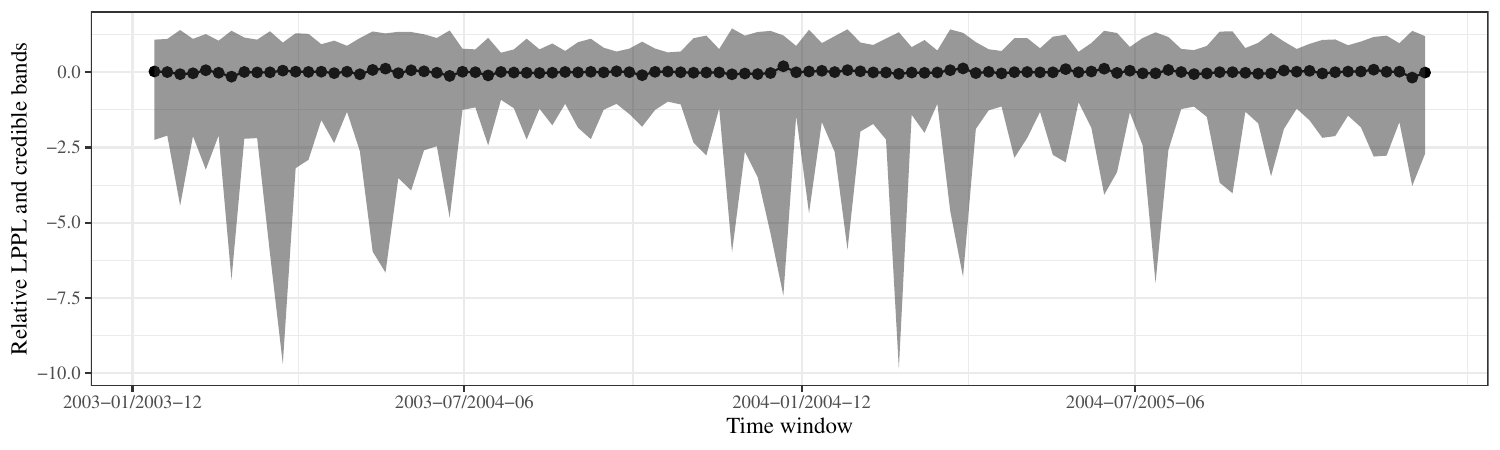}
\caption{Real data exercise.
  Log posterior predictive likelihood (LPPL) under the shrinkage prior on $\tau_j$ with sufficient variance identification minus LPPL without variance identification.
  Dots represent the difference in log posterior means and the gray intervals are the 5th to 95th percentiles credible regions of LPPL under the unrestricted prior.}
\label{fig:logliks}
\end{figure}

The top panel of Figure~\ref{fig:nfacs} displays
the shift in
the posterior distribution $p(\nfactrue\mid y_1,\ldots,y_{52})$ \SFS{across time} when switching from the unrestricted scenario to the sufficient scenario under the uniform prior \SFS{on $\tau_j$}.
The bottom panel shows the same for the shrinkage prior.
For instance, the blue triangle at the ``2003-07/2004-06'' label in the ``Uniform prior'' facet at $\nfactrue=4$ denotes approximately 0.15, which means that the posterior probability of $\nfactrue=4$ is 15 percentage points higher under the sufficient scenario than under the unrestricted scenario.
Probabilities of large $\nfactrue$ are generally reduced, and the probabilities of small $\nfactrue$ are increased.
We do not report results for the necessary scenario here, but the image is similar.
The sea of downward-pointing triangles \SFS{lies} %
above the sea of upward-pointing triangles in both panels, which indicates that the sufficient scenario consistently reduces the estimated number of factors $\nfactrue$ compared to the unrestricted scenario.

\begin{figure}[t]
  \centering
\includegraphics[width=\linewidth]{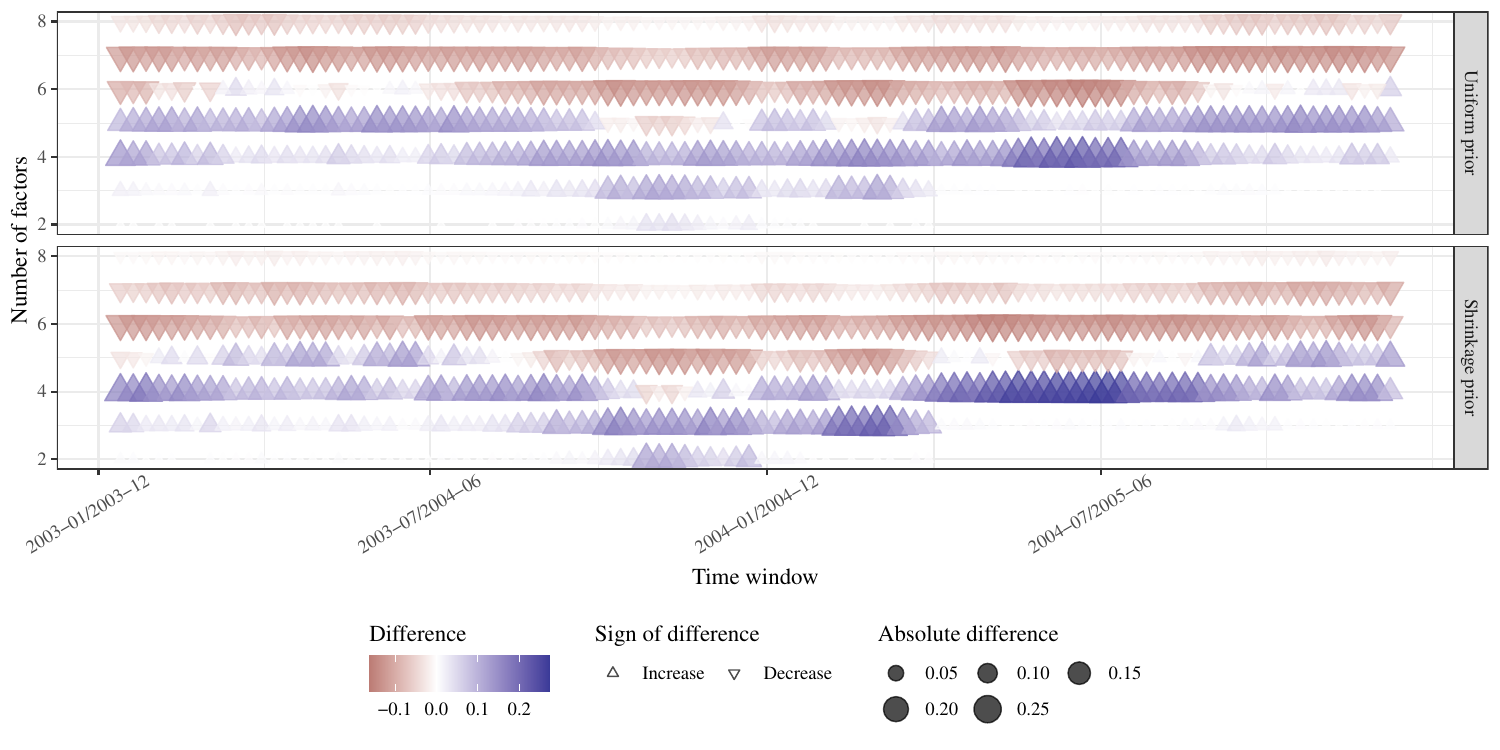}
\caption{Real data exercise.
  Shift in the posterior distributions of the number of factors when switching from the unrestricted scenario to the sufficient scenario.
  The posterior distributions run through time for the uniform prior (top) and the shrinkage prior (bottom).
  The size, color and orientation of the triangles represent the change in posterior probability.
  The triangles are transparent, and only every fourth period is shown for improved legibility.}
\label{fig:nfacs}
\end{figure}

In our experience, the share of variance identified matrices increases in the posterior sample with more shrinkage, and this is reflected in Figure~\ref{fig:probs}, which shows the posterior proportion of variance identified $\binarymatrix$ matrices under the two prior specifications.
The shrinkage prior prefers either close to empty or close to full columns a priori, separately for each column.
In contrast, the uniform prior produces close to half full columns a priori.
This spills over to the posterior distribution for this data set as can be seen from the proportions.
The counting rule $\CRr$ is more likely satisfied with more crowded columns, which results in slightly higher acceptance rates in all time windows.
Rates are mostly between 25\% and 45\%, and the difference between the two priors is consistent but not substantial.
Further investigations not reported here show that the necessary scenario results in a similar increase in the proportion of variance identified matrices as the sufficient scenario does.
Moreover, increasing shrinkage by decrasing $\alpha$ from five to three increases the distance between the two priors in the proportion of variance identified matrices, further supporting the conclusion that column shrinkage is beneficial for variance identification.

\begin{figure}[t]
\includegraphics[width=\linewidth]{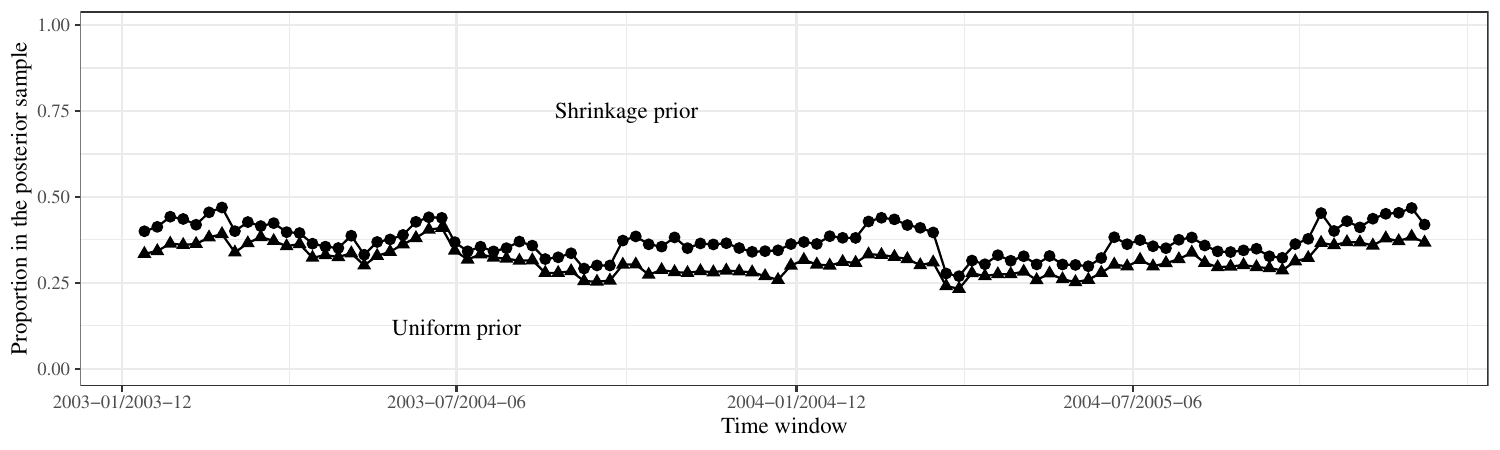}
\caption{Real data exercise.
  Fraction of variance identified draws in the posterior sample under the shrinkage and the uniform prior \SFS{on $\tau_j$}.}
\label{fig:probs}
\end{figure}

Overall, \SFS{both} the simulation study and the real world application consistently show that variance identification reduces the estimated number of factors without affecting the quality of the estimated covariance matrix.
One drawback is increased computational time for the same number of draws, as parts of the MCMC output are discarded, but the \SFS{ensuing} %
reduction in efficiency is \SFS{small} %
compared to the benefits of an improved estimator.

\section{Conclusion}\label{sec:conclusion}

\SFS{In this paper, we studied factor models which are a highly useful technique for dimension reduction in multivariate statistical analysis. To add to the mathematical understanding of these models, we focused  on variance identification to  uniquely identify
the variance decomposition in the factor representation of a covariance matrix. We proved
that a well-known counting rule based on the zero-nonzero pattern of the loading matrix is a sufficient condition for achieving  variance identification. The proof relied on  connecting factor analysis with some classical elements from graph and network theory which to our knowledge has not been exploited so far.}

\SFS{To enhance the relevance of this mathematical insight for practical factor analysis, we provide a computationally efficient algorithm for verifying the counting rule that again relies on results in graph and network theory.}  Our methodology is illustrated for simulated as well as real data in the context of post-processing posterior draws in Bayesian sparse factor analysis.
\SFS{As a main conclusion we find that certain inference tasks in factor analysis such as a predictive analysis are robust to whether
posterior draws are variance identified, while others inference tasks such as identifying number of factors may be hugely impacted
by the presence of unidentified posterior draws.}

\section*{Acknowledgments}

We thank the Editor, Associate Editor and referees.

\section*{Appendix\label{sec:appendix}}
\setcounter{subsection}{0}
\renewcommand{\thesubsection}{\Alph{subsection}}

\subsection{Example of a Variance Identified Model Without the Row Deletion Property}\label{sec:example}

The counting rule is not necessary.
The following sparse space is an example where the counting rule does not hold but the model is generically globally variance identified.
\begin{equation*}
  \facloadmatrix=\begin{pmatrix}
    \facload{11} & 0            & 0 \\
    \facload{21} & \facload{22} & 0 \\
    0            & \facload{32} & \facload{33} \\
    \facload{41} & 0            & \facload{43} \\
    0            & \facload{52} & 0 \\
    0            & 0            & \facload{63}
  \end{pmatrix}, \quad
  \Omega=
  \begin{pmatrix}
    v_1 & \cdot & \cdot & \cdot & \cdot & \cdot \\
    c_{21} & v_2 & \cdot & \cdot & \cdot & \cdot \\
    0 & c_{32} & v_3 & \cdot & \cdot & \cdot \\
    c_{41} & c_{42} & c_{43} & v_4 & \cdot & \cdot \\
    0 & c_{52} & c_{53} & 0 & v_5 & \cdot \\
    0 & 0 & c_{63} & c_{64} & 0 & v_6
  \end{pmatrix}.
\end{equation*}
To see this, observe that all factor loadings can generically be computed (up to sign switches in each column) from the lower triangular elements of $\Omega$, e.g., $\facload{11}=\sqrt{c_{21}c_{41}/c_{42}}$, and then all elements of $\Sigmaeps$ can generically be computed given the factor loadings and the diagonal of $\Omega$.

\subsection{MCMC Algorithm} \label{app:mcmc}

In this section, we provide details on the sampling algorithm that we employ for the numerical illustrations.
To keep the presentation concise, we denote by $j{:}k$ the sequence $(j, j+1, \ldots, k)$; if $j>k$, then $j{:}k$ is an empty sequence.
Furthermore, for any $n\times m$-dimensional matrix $X$, $X_{.,j{:}k}$ denotes the submatrix of $X$ consisting of its $j$th to $k$th columns.
Similarly, $X_{.,j}$ is just the $j$th column and $X_{j,.}$ is the $j$th row of $X$.
Finally, with a slight abuse of notation for the data vector $y$, $y_{i,t}$ denotes the $i$th element of the column vector $y_t$.

\begin{algorithm}
	\caption{MCMC Sampler for Sparse Factor Analysis}
	\label{alg:mcmc}
	\begin{algorithmic}[1]
		\Procedure{SampleSFA}{$\{y_t\}$, $\theta^{(0)}$, $M$, $a_0$, $b_0$, $c_0$, $C_0$}
		\For{$l$ \textbf{in} $\{1,\ldots,M\}$} \Comment{Main sampling loop}
		\For{$j$ \textbf{in} $\{1,\ldots,\rmax\}$} \Comment{Sample \SFS{$\delta_{:,j}$ and $\tau_j$} column-wise}
		\For{$i$ \textbf{in} $\{1,\ldots,m\}$} \Comment{Row-wise elements $\delta_{ij}$ %
  are sampled independently}
    \State Draw $\delta_{ij}^{(l)} \sim
    p(\delta_{ij}\mid\delta_{.,1{:}(j-1)}^{(l)},\delta_{.,(j+1){:}\rmax}^{(l-1)},\tau_j^{(l-1)},\{f_t^{(l-1)}\},\{y_t\})$
		\EndFor

		\State Draw $\tau_j^{(l)} \sim p(\tau_j\mid\delta_{.,j}^{(l)})$
		\EndFor

		\For{$i$ \textbf{in} $\{1,\ldots,m\}$} \Comment{Sample \SFS{$\beta_{i,.}$ and $\sigma^2_i$} row-wise}
		\State Draw $(\sigma_i^2)^{(l)} \sim p(\sigma_i^2\mid\delta_{i,.}^{(l)},\{f_t^{(l-1)}\},\{y_{i,t}\})$
		\State Draw $\beta_{i,.}^{(l)} \sim p(\beta_{i,.}\mid(\sigma_i^2)^{(l)},\delta_{i,.}^{(l)},\{f_t^{(l-1)}\},\{y_{i,t}\})$
		\EndFor

		\For{$t$ \textbf{in} $\{1,\ldots,T\}$} \Comment{Sample $f_t$ separately for each time point}
		\State Draw $f_t^{(l)} \sim p(f_t\mid\{(\sigma_i^2)^{(l)}\},\beta^{(l)},y_t)$
		\EndFor
		\EndFor
		\State \textbf{return} $\{\theta^{(l)}\}$ for $l=1,\ldots,M$
		\EndProcedure
	\end{algorithmic}
\end{algorithm}

\SFS{Algorithm~\ref{alg:mcmc} is a simplified version of the MCMC algorithm by~\cite{Fruehwirth2024Sparse}, adjusted to unrestricted loading matrices.
In particular,
the conditional posterior distributions
in lines (5),   (10-11), and (14) of Algorithm~\ref{alg:mcmc} are based on steps (Da),  (P), and (F), respectively, in their notation.
Line~(7) is a modified version of their step~(H),
taking into account that no restriction is imposed on $\beta$ to resolve rotational invariance, i.e.:
\begin{equation*}
\tau_j\mid\delta_{.,j}^{(l)} \sim
B(a_0 + d_j, b_0+\dimy - d_j), \quad \text{where}
\quad d_j=\sum_{i=1}^\dimy \delta_{ij}.
\end{equation*}
$M$ denotes the total number of MCMC draws.}

\bibliographystyle{myjmva}
\bibliography{exported_references}

\end{document}